\documentclass[final,5p,times,twocolumn]{elsarticle}

\usepackage{amssymb}
\usepackage{amsmath}
\usepackage{graphicx}%
\usepackage{multirow}%
\usepackage{amsmath,amssymb,amsfonts}%
\usepackage{amsthm}%

\newtheorem{theorem}{Theorem}
\newtheorem{lemma}{Lemma}
\newtheorem*{proof}{Proof}\renewenvironment{proof}[1][\proofname]{%
	\par\noindent\textnormal{\bfseries #1.} \ignorespaces
}{%
	\hfill \qedsymbol
}

\newtheorem{definition}{Definition}
\usepackage{mathrsfs}%
\usepackage{textcomp}%
\usepackage{manyfoot}%
\usepackage{booktabs}%
\usepackage{algorithm}%
\usepackage{algorithmicx}%
\usepackage{algpseudocode}%
\usepackage{listings}%
\usepackage{framed}%
\usepackage{comment}%
\usepackage{arydshln}  
\usepackage{makecell}
\usepackage{threeparttable}
\usepackage{bbding}
\usepackage[table,xcdraw]{xcolor}
\usepackage[section]{placeins}
\usepackage{hyperref}
\hypersetup{
	colorlinks=true, 
	linkcolor=blue, 
	citecolor=blue, 
	urlcolor=blue  
}

\begin{document}

	\begin{frontmatter}

		\title{Lattice-Based Dynamic $k$-Times Anonymous Authentication with Attribute-Based Credentials}
		
		\author[label1]{Junjie Song} 
		\ead{junjiesong@seu.edu.cn}
		\author[label1,label2]{Jinguang Han\corref{corresponding}} 
		\ead{jghan@seu.edu.cn}
		\author[label3]{Man Ho Au}
		\ead{mhaau@polyu.edu.hk}
		\author[label4]{Rupeng Yang}
		\ead{rupengy@uow.edu.au}
		\author[label1]{Chao Sun}
		\ead{sunchaomt@seu.edu.cn}
		\affiliation[label1]{organization={School of Cyber Science and Engineering, Southeast University},
			city={Nanjing 211189},
			country={China}
		}
		\affiliation[label2]{
			organization={Wuxi Campus, Southeast University},
			city={Wuxi 214125},
			country={China}
		}
		\affiliation[label3]{
			organization={Department of Computing, The Hong Kong Polytechnic University},
			city={Hong Kong},
			country={China}
		}
		\affiliation[label4]{
			organization={School of Computing and Information Technology, University of Wollongong},
			city={Wollongong 2522},
			country={Australia}
		}
		
		\cortext[corresponding]{Corresponding author.}
		\begin{abstract}
			With the development of Internet, privacy has become a primary concern of users. Anonymous authentication plays an important role in privacy-preserving systems. 
			A $k$-times anonymous authentication ($k$-TAA) scheme allows a group member to anonymously authenticate to the application provider up to $k$ times.
			Considering quantum computing attacks, lattice-based $k$-TAA was introduced. However, existing schemes neither support the dynamic granting and revocation of users nor enable users to control the release of their attributes.
			In this paper, we construct the first lattice-based dynamic $k$-TAA, which offers limited times anonymous authentication, dynamic member management, attribute-based authentication, and post-quantum security.
			We present a concrete construction, and reduce its security to standard complexity assumptions. 
			Notably, compared with existing lattice-based $k$-TAA, our scheme is efficient in terms of communication cost.
		\end{abstract}
		
		
		
		
	\end{frontmatter}
	
		
	\section{Introduction}\label{sec1}
	Many applications, such as electronic voting \cite{Chaum81,Neff01,ChillottiGGI16,FarzaliyevPSW25}, electronic cash \cite{Chaum82,tcs_VaradharajanNM99,CamenischHL05,asiacrypt_LLNW17,DeoLNS20}, and trial browsing of content \cite{k-TAA_TFS04,dynamic_k-taa_NS05}, require the verification of users' identities while protecting their privacy. Anonymous authentication enables authorized users to authenticate with application providers (AP) while their real identities will not be disclosed. Hence, anonymous authentication can implement authentication and privacy protection. However, traditional anonymous authentication has two main issues: (1) it is difficult to trace dishonest users; (2) it cannot limit the authentication times.
	
	To fix the problems mentioned above,  Teranishi et al. \cite{k-TAA_TFS04} first proposed the definition of $k$-times anonymous authentication ($k$-TAA).
	In $k$-TAA, an authorized user is able to anonymously authenticate to AP at most $k$ times; otherwise, the user can be de-anonymized and identified. Due to its metrics, $k$-TAA has been applied to e-voting  system \cite{Chaum81,Neff01,ChillottiGGI16,FarzaliyevPSW25}, e-cash system \cite{Chaum82,tcs_VaradharajanNM99,CamenischHL05,asiacrypt_LLNW17,DeoLNS20} and system supporting trial browsing of content \cite{k-TAA_TFS04,dynamic_k-taa_NS05}. to protect users' privacy and limit the authentication times.
	
	Considering revocation issues in $k$-TAA, Nguyen et al. \cite{dynamic_k-taa_NS05} proposed the first dynamic $k$-TAA where an AP can revoke a user even he/she accesses a service less than the allowed times.
	In dynamic $k$-TAA, each AP can grant and revoke users dynamically and independently. For instance, when offering trial browsing services, an AP might prefer to grant access to users who have favorable profiles, or it might charge some fee for access group membership.
	
	Nevertheless, existing dynamic $k$-TAA schemes cannot resist quantum computing attacks.
	Although post-quantum $k$-TAA has been proposed \cite{Rupeng_wPRF}, it does not support dynamic grant and revocation, and can not control the release of attributes.
	Therefore, it is unsuitable to environments that require flexible user management and fine-grain authentication. In this paper, we design the first post-quantum dynamic $k$-TAA with attribute-based authentication to address the problem mentioned above.
	
	\begin{table*}[htbp]
		\centering
		\caption{\centering Comparison of existing $k$-TAAs}
		\label{tab:comparisons}
		\begin{tabular}{lllll}
			\toprule
			{Schemes}  & \makecell[l]{Post-quantum\\security}  & \makecell[l]{Dynamic\\management} & \makecell[l]{Attribute-based\\ authentication} & {Assumptions}\\
			\midrule
			\cite{k-TAA_TFS04}      & {\XSolidBrush}  & {\XSolidBrush} & {\XSolidBrush}  &   S-RSA, DDH    \\
			\cite{periodic_k-TAA}       & {\XSolidBrush}  & {\XSolidBrush} & {\XSolidBrush}  &   S-RSA, $y$-DDHI / SDDHI   \\
			\cite{tifs_YangXSYLPD24}    & {\XSolidBrush}  & {\XSolidBrush} & {\Checkmark}  &   DL, DDH, LRSW  \\
			\cite{LianCML15}	  & {\XSolidBrush}  & {\XSolidBrush}  & {\XSolidBrush}  &    S-RSA, DDH, $y$-DHI  \\
			\cite{dynamic_k-taa_NS05}      & {\XSolidBrush}  & {\Checkmark} & {\XSolidBrush}  &  $q$-SDH, DBDH   \\
			\cite{Allen_k_TAA}          & {\XSolidBrush}  & {\Checkmark} & {\XSolidBrush}   &   DDH, $q$-SDH, $y$-DDHI    \\
			\cite{tc_YuenLAHSZ15}       & {\XSolidBrush}  & {\Checkmark}  & {\Checkmark}  &  $y$-DBDHI    \\
			\cite{tifs_ChaterjeeMC21}    & {\XSolidBrush}  & {\XSolidBrush} & {\XSolidBrush}  &  DDH, D-Uniqueness of PUF   \\
			\cite{tdsc_HuangSGWZH22}   & {\XSolidBrush}  & {\XSolidBrush}& {\XSolidBrush}  &  $q$-SDH, $y$-CBDH, $y$-DBDH     \\
			\cite{Rupeng_wPRF}     & {\Checkmark}    & {\XSolidBrush}& {\XSolidBrush}  &     LWE, LWR, SIS   \\
			\cite{asiacrypt_LLNW17}     & {\Checkmark}    & {\XSolidBrush}& {\XSolidBrush}  &     LWE, LWR, SIS   \\
			\cite{e-cash_DLNS20}     & {\Checkmark}    & {\XSolidBrush}& {\XSolidBrush}  &     LWE, LWR, SIS   \\
			\midrule
			\rowcolor[HTML]{E0E0E0} {Ours}        & {\Checkmark}  & {\Checkmark} & {\Checkmark} &   LWE, LWR, SIS \\
			\bottomrule
		\end{tabular}
	\end{table*}

	\subsection{Related Works}
	After Teranishi et al. \cite{k-TAA_TFS04} introduced the first $k$-TAA scheme, $k$-TAA schemes with different features were proposed.
	Chaterjee et al. \cite{tifs_ChaterjeeMC21} presented a $k$-TAA scheme based on physically unclonable functions, which is suitable for trusted platform modules (TPM). To improve the efficiency of $k$-TAA, Huang et al. \cite{tdsc_HuangSGWZH22} proposed an efficient $k$-TAA, which has constant computation cost and communication cost, making it practical for real-world pay-as-you-go cloud computing.
	
	Camenisch et al. \cite{periodic_k-TAA} first presented a periodic $k$-TAA scheme, where a registered user is allowed to anonymously authenticate himself/herself $k$ times in each time period. Lian et al. \cite{LianCML15} presented a more efficient periodic $k$-TAA, and their scheme achieves the ``least additional costs'' for revocation by leaking the dishonest user's secret parameter through a special zero-knowledge proof. Later, Yang et al. \cite{tifs_YangXSYLPD24} proposed a periodic fine-grained $k$-TAA, which supports selective disclosure of attributes and does not require the time consuming pairing operation. Therefore, this scheme provides flexible access control and is efficient.
	
	However, all these schemes cannot support dynamic member management. To support dynamic grant and revocation, Nguyen and Safavi-Naini \cite{dynamic_k-taa_NS05} introduced the first dynamic $k$-TAA, where AP can dynamically grant or revoke users. To reduce the communication costs of dynamic $k$-TAA, Au et al. \cite{Allen_k_TAA} proposed a constant-size dynamic $k$-TAA, where the proof cost remains constant instead of linearly with $k$. Subsequently, Yuen et al. \cite{tc_YuenLAHSZ15} designed an attribute-based dynamic $k$-TAA scheme that supports flexible access control policies based on various attributes.

	The above schemes show various advantages, but cannot resist quantum computing attacks. Considering this problem, Yang et al. \cite{Rupeng_wPRF} presented the first post-quantum $k$-TAA scheme. They formalized a new primitive named weak pseudorandom function (wPRF) with efficient protocols to ensure the ``$k$-times'' authentication, and developed an extended abstract Stern's protocol, which provides a general zero-knowledge argument of knowledge (ZKAoK) framework for the statements required in their $k$-TAA.
	However, this scheme is an ordinary $k$-TAA, so AP lacks the ability to grant or revoke access to users dynamically and to support attribute-based authentication. 
	Meanwhile, they just showed the theoretical possibility of a lattice-based $k$-TAA without providing an concrete instantiation, and it is actually inefficient.
	
	Additionally, e-cash \cite{CamenischHL05} can be viewed as a special form of $k$-TAA, 
	where the authentication tokens take the form of e-coins that can be used up to $k$ times 
	in the context of electronic payments. Libert et al. \cite{asiacrypt_LLNW17} proposed the first lattice-based e-cash scheme, and Yang et al. \cite{YAZ+19_ZKP_on_Z} later introduced a more efficient construction using a newly developed ZKAoK framework. Subsequently, Deo et al. \cite{e-cash_DLNS20} showed that both prior e-cash schemes fail to provably achieve exculpability, a property that ensures no entity can falsely accuse an honest user of double spending, and are therefore insecure. They further proposed the first concrete lattice-based e-cash systems that eliminate this security limitation. 
	Despite these advances, existing lattice-based e-cash schemes, as specific instantiations of $k$-TAA, still do not support dynamic management and attribute-based authentication.
	
	In this paper, we propose the first post-quantum dynamic $k$-TAA scheme with attribute-based credentials based on lattices. A detailed comparison between our scheme and existing $k$-TAA is presented in Table \ref{tab:comparisons}. 
	
	\subsection{Our Contributions}  
	Our new scheme provides the following interesting properties and features:
	
	\begin{enumerate}
		\item \textbf{Anonymity.} 
		A user's identity cannot be publicly identified if and only if he/she authenticates no more than $k$ times.
		\item \textbf{Accountability.}
		A user's identity can be publicly identified if and only if he/she authenticates more than $k$ times.
		\item \textbf{Exculpability.} 
		No honest user can be falsely accused of authenticating with the same honest AP more than $k$ times.
		\item \textbf{Dynamicity.} 
		The AP can grant and revoke users dynamically and independently, even if the users authenticate less than $k$ times.
		\item \textbf{Efficiency.} Our scheme achieves higher efficiency compared to existing lattice-based $k$-TAA in terms of communication costs.
		\item \textbf{Attribute-Based Credentials.}
		The authorized users are allowed to selectively disclose attributes during authentication by using the credentials from the GM.
		\item \textbf{Post-quantum Security.} 
		The proposed scheme relies on standard lattice-based hardness assumptions and therefore achieves post-quantum security.
	\end{enumerate}
	
	Our contributions are summarized as follows: (1) we build a concrete construction of our dynamic $k$-TAA scheme; (2) our scheme is instantiated and compared with related schemes; (3) the security of our scheme is formally proven. The novelty is to protect users' privacy and implement flexible revocation.
	
	\section{Preliminary}
	\subsection{Notations}
	Table \ref{tab:notation} summarizes the basic notations used in this paper.
	
	\begin{table*}[htbp]
		\centering
		\caption{\centering Basic notations used in this paper}
		\label{tab:notation}
		\scalebox{0.94}{
			\begin{threeparttable}
				\begin{tabular}{l@{\hspace{18pt}}l}
					\toprule
					Notation & Description\\
					\midrule
					$[a,b]$ & The set $\{k\in\mathbb{Z}:a\leq k\leq b\}$ for $a\leq b$.\\
					$[c]$ & The set $\{k\in\mathbb{Z}:0\leq k\leq c-1\}$ for $1\leq c$.\\
					$\mathbb{Z}_q$ & For an integer $q>0$, define $\mathbb{Z}_q = \mathbb{Z}/q\mathbb{Z}$.\\
					$\mathbf{A}$ / $\mathbf{v}$ & The matrix / the column vector  in $\mathbb{Z}_q$.\\
					$\mathbf{v}[i]$ & The $i$-th entry of  $\mathbf{v}$.\\
					$x\leftarrow\chi$ &$x$ is sampled according to the distribution $\chi$.\\
					$s\xleftarrow{\$}\mathcal{S}$ & $s$ is  uniformly sampled from the set $\mathcal{S}$.\\						
					$\lfloor\cdot\rceil$ & The rounding function.\\
					$\lfloor\cdot\rfloor$ / $\lceil\cdot\rceil$ & The floor function / the ceiling function.\\
					${\lfloor \cdot \rceil}_{p}$ & The map from $\mathbb{Z}_{q}$ to $\mathbb{Z}_{p}$, ${\lfloor x \rceil}_{p}={\lfloor \frac{p}{q}\cdot x \rfloor}$, where $q \geq p \geq 2$.\\
					
					$\|\mathbf{v}\|_\infty$ & The infinity ($\ell_\infty$) norm of $\mathbf{v}$.\\
					$\|\mathbf{v}\|_2$ & The Euclidean ($\ell_2$) norm of $\mathbf{v}$.\\
					\multirow{2}{*}{$\mathbf{I}_n \otimes \mathbf{g}$} & The block-diagonal matrix $\mathrm{diag}(\mathbf{g},\ldots,\mathbf{g}) \in \mathbb{Z}_q^{n \times n\ell}$, \\
					& \qquad constructed from the $n \times n$ identity matrix $\mathbf{I}_n$ and a vector $\mathbf{g} \in \mathbb{Z}_q^\ell$. \\
					
					\multirow{2}{*}{$\vartheta_{j}^{N}$} & For $\mathbf{y} = [\mathbf{y}_0^\top \| \cdots \| \mathbf{y}_{N-1}^\top]^\top \in \mathbb{Z}^{\ell N}$ with 
					$\mathbf{y}_0,\ldots,\mathbf{y}_{N-1} \in \mathbb{Z}^\ell$, \\
					& \qquad define $\vartheta_{j}^{N}$ $(j \in [N])$ as the vector 
					satisfying $\mathbf{y}_j = 1^\ell$ and $\mathbf{y}_i = 0^\ell$ for all $i \neq j$. \\
					
					$\mathsf{bin}(\cdot)$ & The element-wise binary decomposition of vector.\\
					$\mathsf{vdec}(\cdot)$ & The mapping $\mathsf{vdec} : \mathbb{Z}_{q}^{n} \to [0,2^\iota-1]^{nk'}$, defined by 
					$\mathsf{vdec}(\mathbf{a}) = \bar{\mathbf{a}}$. \tnote{$\dagger$}\\
					\multirow{2}{*}{$\mathsf{M2V}(\cdot)$} & The column vector concatenation of matrix, $\mathbf{v} = \mathsf{M2V}(\mathbf{M}) = (\mathbf{M}[1]^{\top}\|\cdots\|\mathbf{M}[n]^{\top})^{\top}$, \\
					&\qquad 
					where $\mathbf{M}[j]_{j\in[1,n]}$ denote columns of $\mathbf{M}\in\mathbb{Z}^{m\times n}$. \\
					$negl(n)$ & A function that is negligible in $n$.\\
					$poly(n)$ & A function that is polynomial in $n$.\\
					\bottomrule
				\end{tabular}
				\begin{tablenotes}
					\item[$\dagger$] $\forall i\in[1,n]$, $\mathbf{a}[i]=\sum^{k'}_{j=1}\left(\left(2^\iota\right)^{j-1}\cdot k'\cdot\bar{\mathbf{a}}[\left(i-1\right) +j]\right)$, where $k'=\frac{\log{q}}{\iota}$.
				\end{tablenotes}
			\end{threeparttable}
		}
	\end{table*}
	
	\subsection{Lattices}
	A lattice $\Lambda$ is defined as:		
	$$\Lambda = \mathcal{L}(\mathbf{b}_1, \ldots, \mathbf{b}_m) = \left\{ \sum_{i=1}^{m} z_i \mathbf{b}_i : z_i \in \mathbb{Z} \right\},$$		
	where the vectors $(\mathbf{b}_1, \ldots, \mathbf{b}_m)$ form a basis for the lattice. In addition, if $n=m$, the lattice $\Lambda$ is said to be full-rank.
	
	\begin{definition}
		Given positive integers $n,m,q$ with $q\geq2$, and a matrix $\mathbf{A} \in \mathbb{Z}_q^{n \times m}$, we define the q-ary lattice consisting of the vectors which are orthogonal to $\mathbf{A}$, more concretely,
		$$\Lambda^\perp(\mathbf{A}) = \{ \mathbf{e} \in \mathbb{Z}^m : \mathbf{Ae} = \mathbf{0} \mod q \},$$
		and for any vector $\mathbf{u} \in \mathbb{Z}_q^n$  generated by the linear combination of columns of $\mathbf{A}$, we define a lattice coset consisting of the vectors defined as follows:
		$$\Lambda^\mathbf{u}(\mathbf{A}) = \{ \mathbf{e} \in \mathbb{Z}^m : \mathbf{Ae} = \mathbf{u} \mod q \}.$$
	\end{definition}
	
	\begin{definition}
		Let $\mathbf{c}\in\mathbb{R}^n$ be a center and $\sigma>0$ be a standard deviation. \textbf{The Gaussian function} centered at $\mathbf{c}$ with parameter $\sigma$ is defined by:
		$$\forall \mathbf{x} \in \mathbb{R}^n, \rho_{\sigma, \mathbf{c}}(\mathbf{x}) = \exp \left( -\frac{\|\mathbf{x} - \mathbf{c}\|^2}{2\sigma^2} \right).$$
	\end{definition}
	
	\begin{definition}
		Let $\mathbf{c}\in\mathbb{R}^n$ be a center and $\sigma>0$ be a standard deviation. \textbf{The discrete Gaussian distribution} over an $n$-dimensional lattice $\Lambda$ centered at $\mathbf{c}$ with parameter $\sigma$ is defined by:
		
		$$\forall \mathbf{x} \in \Lambda, \mathcal{D}_{\Lambda, \sigma, \mathbf{c}} = \frac{\rho_{\sigma, \mathbf{c}}(\mathbf{x})}{\sum_{\mathbf{x} \in \Lambda} \rho_{\sigma, \mathbf{c}}(\mathbf{x})}.$$
		
	\end{definition}
	
	\subsection{Lattice Trapdoors}
	By \textbf{TrapGen} and \textbf{SamplePre}, we denote the trapdoor generation algorithm that generates a random lattice together with a trapdoor corresponding to the lattice, and the pre-sampling algorithm that samples short vectors within a given lattice coset, respectively.
	\begin{theorem}[\cite{MP12}]
		Given positive integers $n,m,q$ with $q\geq2$ and $m=\mathcal{O}(n\log q)$, there exists 
		a probabilistic polynomial time (PPT) algorithm $\mathbf{TrapGen}(1^n, 1^m, q)$ outputs a pair $(\mathbf{A}, \mathbf{T})$, where $\mathbf{A} \in \mathbb{Z}_q^{n \times m}$ is within $2^{-\Omega(n)}$ statistical distance of a uniform matrix, and $\mathbf{T} \in \mathbb{Z}_q^{m \times m}$ is a basis for the lattice $\Lambda^\perp(\mathbf{A})$.
	\end{theorem}
	
	\begin{theorem}[\cite{MP12}]
		Given positive integers $n,m,q$ with $q\geq2$ and $m=$ $O(n\log q)$, there exists 
		a PPT algorithm $\mathbf{SamplePre}(\mathbf{A},\mathbf{T},\mathbf{u}, \sigma$), which takes as input a matrix $\mathbf{A}\in\mathbb{Z}_q^{n\times m}$, a short basis $\mathbf{T}\in\mathbb{Z}_q^{m\times m}$, a gaussian parameter $\sigma\geq\omega(\sqrt{\log n})\cdot\|\tilde{\mathbf{T}}\|$ and a vector $\mathbf{u}\in\mathbb{Z}_q^n$, then outputs a vector $\mathbf{e}\in\mathbb{Z}^m$, and the distribution of $\mathbf{e}$ is statistically close to $\mathcal{D}_{{\Lambda ^\mathrm{u}} ( \mathbf{A} ) , \sigma }$.
		
	\end{theorem}
	
	\subsection{Complexity Assumptions}
	In this work, we consider the following well-known complexity assumptions on lattices.
	\begin{definition}[Short Integer Solution Assumption (${SIS}_{n,m,q,\beta}$) \cite{Ajtai96}]
		Let $\mathbf{A}\in\mathbb{Z}_q^{n\times(m-n)}$ be a random matrix and let $\mathbf{I}_n$ be a $n$-th identity matrix.
		We say that the ${SIS}_{n,m,q,\beta}$ assumption holds if all PPT adversaries $\mathcal{A}$ can find a solution $\mathbf{z}$ that is short (i.e., $0\leq\|\mathbf{z}\|\leq\beta$) and satisfies  $[\,\mathbf{I}_{n}\mid\mathbf{A}\,]\cdot\mathbf{z}=0$ with negligible probability $negl(\lambda)$, namely, 
		\[
		Adv_\mathcal{A}^{SIS}=
		\Pr\left[
		\begin{aligned}
			&0\le\|\mathbf{z}\|\le\beta,\\
			&[\,\mathbf{I}_{n}\mid\mathbf{A}\,]\cdot\mathbf{z}=0
		\end{aligned}
		\;\middle|\;
		\mathbf{z}\leftarrow\mathcal{A}(\mathbf{A})
		\right]<{negl}(\lambda).
		\]
	\end{definition}
	
	Since the difficulty of the SIS assumption is typically determined only by the parameters $n,q,\beta$ (under the assumption that $m$ is sufficiently large), we adopt the simplified notation ${SIS}_{n,q,\beta}$ in place of ${SIS}_{n,m,q,\beta}$ throughout this work.
	
	\begin{definition}[Learning with Errors Assumption (${LWE}_{n,m,q,\chi}$) \cite{Regev05}]
		Let $\mathbf{A}\in\mathbb{Z}_q^{(m-n)\times n}$ be a random matrix and let $\chi$ be a probability distribution over the elements of $\mathbb{Z}_q$.
		We say that the ${LWE}_{n,m,q,\chi}$ assumption holds if all PPT adversaries $\mathcal{A}$ distinguish a generated vector $\mathbf{A}\cdot\mathbf{s}+\mathbf{e}$ from a random vector $\mathbf{v}$ with negligible probability $negl(\lambda)$, namely,
		$$\begin{aligned}
			Adv_\mathcal{A}^{LWE}=\big|\Pr[1 \leftarrow &\mathcal{A}(\mathbf{A}, \mathbf{A\cdot s} + \mathbf{e})]\\
			&- \Pr [1 \leftarrow \mathcal{A}(\mathbf{A}, \mathbf{v})]\:\big|
			<negl(\lambda),
		\end{aligned}
		$$
		where $\mathbf{s}\leftarrow\chi^n$ ,$\mathbf{e}\leftarrow\chi^{m-n}$, and $\mathbf{v}\xleftarrow{\$}\mathbb{Z}_q^{m-n}$.
	\end{definition}
	In the case $\chi$ is defined as a discrete Gaussian distribution with standard deviation $\sigma$, ${LWE}_{n,m,q,\chi}$ can be denoted as ${LWE}_{n, m, q, \alpha }$, with the parameter $\alpha$ defined as $\alpha =\sigma\cdot\sqrt{2\pi}/q$.
	Besides, we can adopt the simplified notation ${LWE}_{n,q,\alpha}$ in place of ${LWE}_{n,m,q,\alpha}$ throughout this work, since the difficulty of LWE assumption is typically determined merely by the parameters $n,q,\alpha$ (under the assumption that $m$ is sufficiently large).
	
	\begin{definition}[Learning with Rounding Assumption (${LWR}_{n,m,q,p}$) \cite{BanerjeePR12}]
		Let $\mathbf{A}\in\mathbb{Z}_q^{m \times n}$ be a random matrix.
		We say that the ${LWR}_{n,m,q,p}$ assumption holds if all PPT adversaries $\mathcal{A}$ distinguish a generated vector $\lfloor \mathbf{A} \cdot \mathbf{s} \rceil_p$ from a random vector $\mathbf{v}$ with negligible probability $negl(\lambda)$, namely,
		$$\begin{aligned}Adv_\mathcal{A}^{LWR}=\big|
			\Pr[1 \leftarrow &\mathcal{A}(\mathbf{A}, \lfloor \mathbf{A} \cdot \mathbf{s} \rceil_p)] \\
			&- \Pr [1 \leftarrow \mathcal{A}(\mathbf{A}, \mathbf{v})]
			\:\big|<negl(\lambda),\end{aligned}$$
		where $\mathbf{s} \stackrel{\$}{\leftarrow} \mathbb{Z}_q^n$, and $\mathbf{v}\stackrel{\$}{\leftarrow} \lfloor \mathbb{Z}_q^m \rceil_p$.
	\end{definition}
	Also, we adopt the simplified notation ${LWR}_{n,q,p}$ in place of ${LWR}_{n,m,q,p}$ throughout this work, since the difficulty of the LWR assumption is typically determined only by the parameters $n,q,p$ (under the assumption that $m$ is sufficiently large).
	In addition, as demonstrated in \cite{BanerjeePR12}, the ${LWR}_{n,q,p}$ problem is at least as difficult as the ${LWE}_{n,q,\alpha}$ problem, 		
	where the parameter $\alpha$ is defined as $\alpha=\sqrt{\frac{(q/p)^{2}-1}{12}}\cdot\sqrt{2\pi}/q$.
	\subsection{Dynamic $k$-TAA}
	\subsubsection{Syntax}
	Figure \ref{fig:dynamicktaa} presents the framework of our scheme, involving three entities: the application provider (AP), the group manager (GM), and users (U). After U register with the GM, each AP independently announces a limit on how many times a user can authenticate anonymously. A user can authenticate up to this limit; otherwise he/she can be traced and revoked. The AP can also dynamically grant or revoke access, even if a user have not exceeded their authentication limit.
	
	\begin{figure*}[t]
		\centering
		\includegraphics[width=0.8\textwidth]{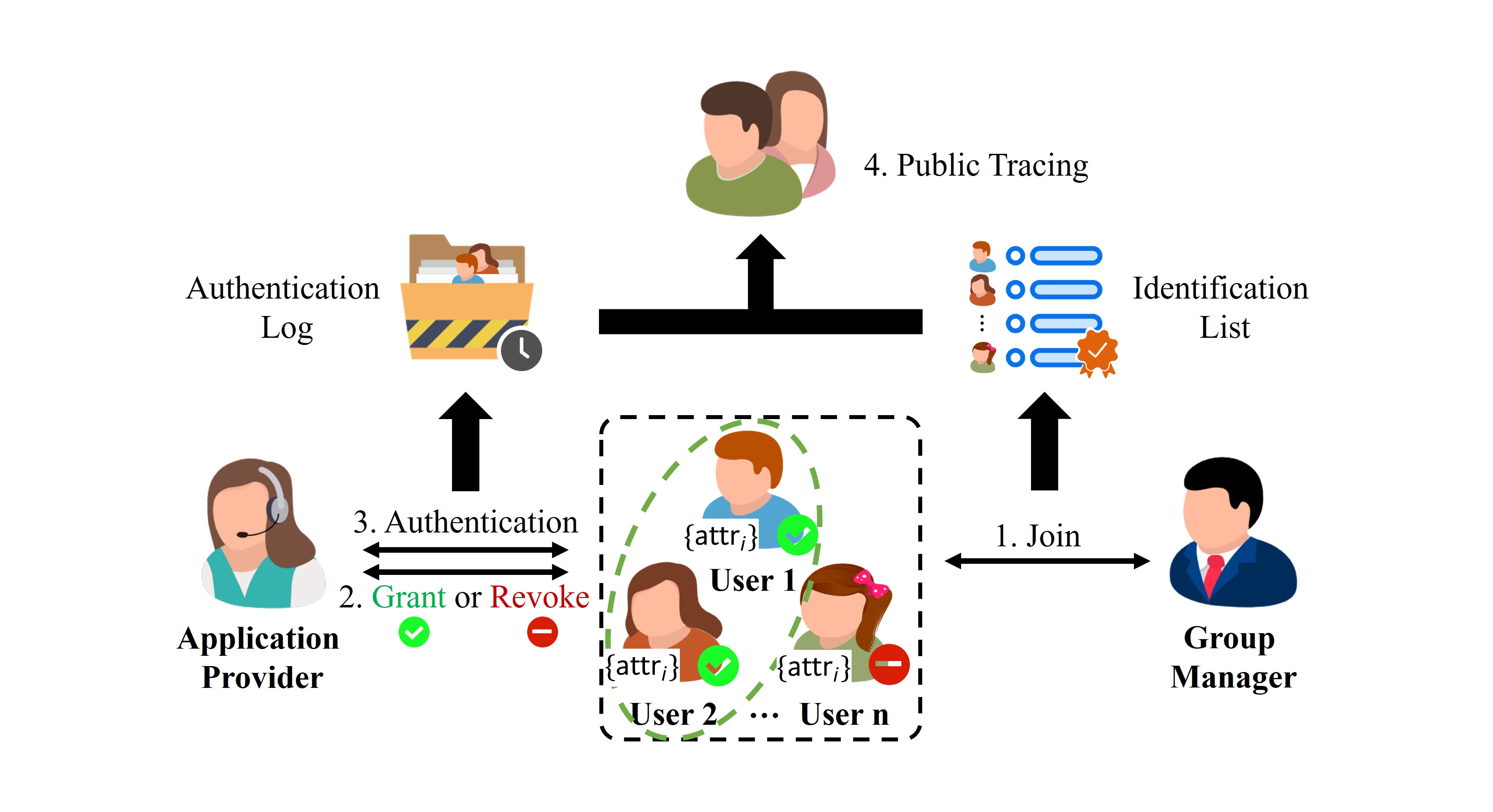}
		\caption{Framework of our dynamic $k$-TAA}
		\label{fig:dynamicktaa}
	\end{figure*}
	
	A dynamic $k$-TAA scheme consists of the following algorithms \cite{dynamic_k-taa_NS05}: 
	\begin{enumerate}
		\item[-] $(usk,upk)\leftarrow\mathbf{UserSetup}(1^{\lambda}).$ This algorithm takes as input a security parameter $1^{\lambda}$, and outputs a pair of public and secret keys ($upk$, $usk$) for each user.
		\item[-] $(gsk,gpk)\leftarrow\mathbf{GMSetup}(1^{\lambda}).$ This algorithm takes as input a security parameter $1^{\lambda}$, and outputs a pair of public and secret keys ($gpk$, $gsk$) for GM.
		\item[-] $((apk,ask),ID,k)\leftarrow\mathbf{APSetup}(1^\lambda).$ This algorithm takes as input a security parameter $1^{\lambda}$, and outputs a pair of public and secret keys $(apk,ask)$ for AP, the identity $ID$ for AP, and the allowed number of times $k$ that an authorized user can access AP's application.
		\item[-] $((mpk,msk),\mathcal{LIST})\leftarrow\mathbf{Join}(U(i,upk,usk)\leftrightarrow GM(gpk,$ $gsk)).$ This interactive algorithm is run between $U$ and $GM$. $U$ takes as input an identity $i$, his/her public-secret key pair ($upk,usk$), and outputs a zero-knowledge proof $\Pi_U$ of the secret key $usk$ corresponding to the public key $upk$. $GM$ takes as input a public-secret key pair ($gpk,gsk$), and outputs a member public-secret key pair ($mpk$, $msk$) for $U$ and an identification list $\mathcal{LIST}$, where $\mathcal{LIST}$ consists of $U$'s identity $i$, his/her public key $upk$ and member public key $mpk$.
		\item[-] $(w,ARC)\leftarrow\mathbf{GrantingAccess}(apk,ask,i,mpk).$ This algorithm takes as input AP's public-secret key pair ($apk,ask$), an authorized user's identity $i$ and his/her member public key $mpk$, and outputs a witness $w$ for $U$ and a public archive $ARC$, where $w$ provides user $i$ an evidence of being granted access, and the public archive $ARC$ is for granted users to update state information.
		\item[-] $ARC\leftarrow\mathbf{RevokingAccess}(apk,ask,i,mpk).$ 
		This algorithm takes as input AP's public-secret key pair ($apk,ask$), a revoked user's identity $i$ and his/her member public key $mpk$, and outputs a public archive $ARC$.
		\item[-] 
		$\left(\left(\mathsf{True}/\mathsf{False}\right),\mathcal{LOG}\right)\leftarrow\mathbf{Authentication}(U(ID,k,i,w,$ $upk,usk,mpk,msk,t)\leftrightarrow AP(ID,k,apk,ask)).$ This interactive algorithm is run between $U$ and $AP$. $U$ takes as input his/her identity $i$, $AP$'s identity $ID$, the allowed number of times $k$, an evidence $w$, his/her public-secret key pair $(upk,usk)$ and the member public-secret key pair $(mpk,msk)$, then outputs a zero-knowledge proof $\Pi$ to prove he/she is an authorized user, and belongs to the access group of the $AP$, with the number of accesses $t$ satisfying $t\leq k$. $AP$ takes as input its own identity $ID$, the allowed number of times $k$, and its public-secret key pair $(apk,ask)$, and outputs $\mathsf{True}$ if $U$ is authenticated successfully; otherwise, outputs $\mathsf{False}$. In addition, AP maintains an authentication log $\mathcal{LOG}$ containing transcripts of all authentication interactions, and this algorithm outputs the $\mathcal{LOG}$.
		\item[-] $(i/GM/\emptyset)\leftarrow\mathbf{PublicTracing}(\mathcal{LIST},\mathcal{LOG}).$ This algorithm
		takes as input the identification list $\mathcal{LIST}$ and the authentication log $\mathcal{LOG}$, and outputs a set consisting of a user's identity $i$, GM or $\emptyset$, which respectively represents ``the user $i$ attempts to access the application for more than $k$ times'', ``the GM cheated'' and ``No illegal entity is recorded in the $\mathcal{LOG}$'' .
	\end{enumerate}

	A dynamic $k$-TAA scheme is \emph{Correct} if:
	{\small\[
		\Pr\hspace{-0mm}\left[\,
		\begin{aligned}
			&\left(\mathsf{True},\mathcal{LOG}\right)\leftarrow\\
			&\;\mathbf{Authentication}(\\
			&\; U(ID,k,i,w,upk,\\
			&\; usk,mpk,msk,t)\leftrightarrow\\
			& \:AP(ID,k,apk,ask))
		\end{aligned}
		\hspace{0.4mm}\middle|\hspace{0.4mm}
		\begin{aligned}
			&(usk,upk)\leftarrow\mathbf{UserSetup}(1^{\lambda}),\\
			&(gsk,gpk)\hspace{0mm}\leftarrow\hspace{0mm}\mathbf{GMSetup}(1^{\lambda}),\hspace{0.5mm}\\
			&(apk,ask,ID,k)\leftarrow\mathbf{APSetup}(1^\lambda),\\
			&(msk,mpk,\mathcal{LIST})\leftarrow\mathbf{Join}(\\
			&\quad U(i,upk,usk)\leftrightarrow GM(gpk,gsk)),\\
			&(w,ARC)\leftarrow\mathbf{GrantingAccess}(\\
			&\qquad\qquad apk,ask,i,mpk).
		\end{aligned}
		\right]\hspace{-0mm}=\hspace{0mm}1
		\]}
	\subsubsection{Security Properties}
	We briefly introduce security properties of dynamic $k$-TAA here, the formal definitions of which are provided in \ref{Formal_Definitions_Requirements}.
	\begin{enumerate}
		\item[-] \textit{D-Anonymity.} 
		A user's identity cannot be publicly identified if and only if he/she authenticates no more than $k$ times, even in case of collusion among other group users, the GM, and the AP.
		\item[-] \textit{D-Detectability.} 
		A user's identity can be publicly identified by the \textbf{PublicTracing} algorithm if and only if he/she authenticates more than $k$ times.
		\item[-] \textit{D-Exculpability.} 
		No honest user can be falsely accused of authenticating with the same honest AP more than $k$ times.
		In addition, if the GM is honest, the output set of \textbf{PublicTracing} algorithm cannot contain GM, even in case of collusion among the users and the AP.
	\end{enumerate}
	\subsection{Building Blocks}		
	The following cryptographic primitives form the core building blocks of our construction: weak pseudorandom function with efficient protocols, non-interactive zero-knowledge arguments of knowledge, signature with efficient protocols, static accumulator scheme, and dynamic accumulator scheme.

	\subsubsection{Weak Pseudorandom Function with Efficient Protocols}

	Our construction employs the weak pseudorandom function (wPRF) with efficient protocols constructed in \cite{Rupeng_wPRF}.
	
	Let $p, q, n, m$ be positive integers, where $p\geq 2$, $\gamma =q/p \in n^{\omega(1)}$ denotes an odd integer, and $m \geq n\cdot (\log q+1)/(\log p-1)$. Two algorithms of wPRF with efficient protocols are as follows.
	\begin{enumerate}
		\item[-]  $\mathbf{s}\leftarrow\mathbf{KeyGen}(1^\lambda)$. On input a security parameter $1^\lambda$, this algorithm samples a vector $\mathbf{s} \xleftarrow{\$} \mathbb{Z}_{q}^{n}$, then it returns $\mathsf{sk}=\mathbf{s}$.
		\item[-] $\mathbf{y}\leftarrow\mathbf{Eval}(\mathbf{s}, \mathbf{A})$. On input a matrix $\mathbf{A} \xleftarrow{\$} \mathbb{Z}_{q}^{m\times n}$, this algorithm computes $\mathbf{y}={\lfloor \mathbf{A}\cdot \mathbf{s} \rceil}_{p}$, then it returns $\mathbf{y}$.
	\end{enumerate}
	
	\begin{theorem}[\cite{Rupeng_wPRF}]
		If the ${LWR}_{n,q,p}$ assumption holds, we additionally require that $m \geq n\cdot (\log q+1)/(\log p-1)$, the scheme in \cite{Rupeng_wPRF} is secure wPRF with weak pseudorandomness and uniqueness properties. 
		Furthermore, the wPRF satisfies the strong uniqueness property when $m \geq 2n\cdot (\log q+1)/(\log p-1)$.
	\end{theorem}
	
	\subsubsection{Zero-Knowledge Arguments of Knowledge}
	A zero-knowledge argument of knowledge (ZKAoK) system \cite{GoldwasserMR85} is run between a prover $\mathcal{P}$ and a verifier $\mathcal{V}$,
	where $\mathcal{P}$ interacts with $\mathcal{V}$ to convince $\mathcal{V}$ that he/she knows a witness for a statement, while exposing no additional information.		
	
	Formally, let $x$ be a statement and $w$ be a corresponding witness.
	Given an NP relation $\mathtt{R} = \{(x, w)\} \in \{0,1\}^* \times \{0,1\}^*$, a ZKAoK protocol for $\mathtt{R}$ satisfies the following properties:		
	\begin{enumerate}
		\item[-] \textbf{Completeness.}  For every $(x,w)\in \mathtt{R}$, $\Pr[\langle \mathcal{P}(x,w), \mathcal{V}(x)\rangle \neq1] \leq \delta_c$, where $\delta_{c}$ denotes the completeness error.
		\item[-] \textbf{Knowledge Soundness.}   
		For every $x$ and any PPT cheating prover $\hat{\mathcal{P}}$, there exists a PPT extractor $\mathcal{E}$ that accesses $\hat{\mathcal{P}}$ in a black-box manner, satisfying that for every $(x,w) \in \mathtt{R}$, $\Pr[\langle \hat{\mathcal{P}}(x,w), \mathcal{V}(x)\rangle = 1\:|\:w\leftarrow\mathcal{E}(x)]>\delta_s+\epsilon$ for a non-negligible probability, where $\delta_{s}$ denotes the soundness error.
		\item[-] \textbf{(Honest Verifier) Zero-Knowledge.}  
		For any $(x,w) \in \mathtt{R}$, 
		a PPT simulator $\mathcal{S}$ can be constructed such that 
		the two distributions below are computationally indistinguishable:
		\begin{itemize} 
			\item[(i)] The simulated transcript generated by $\mathcal{S}(x)$; 
			\item[(ii)] The real transcript viewed by an honest PPT $\mathcal{V}$ during the interaction of $\langle \mathcal{P}(x,w), \mathcal{V}(x)\rangle$.
		\end{itemize}
		%
	\end{enumerate}
	
	Furthermore, we consider non-interactive zero-knowledge proofs of knowledge (NIZKAoK), derived by applying the Fiat-Shamir transformation \cite{FiatS86} to public-coin ZKAoK.
	In addition, it can be transferred to be a signature proof of knowledge scheme if a message is used. By $\text{SPK}\{(x,w) : (x,w) \in \mathtt{R}\}[\mathfrak{m}]$, we denote a signature proof of knowledge protocol, where $\mathfrak{m}$ is a message.
	
	Our construction applies the ZKAoK introduced by Yang et al. \cite{YAZ+19_ZKP_on_Z}, which provides a ZKAoK for the relation $\mathcal{R}^*$ as follows:
	\[\mathcal{R}^{*}\hspace{0mm}=\hspace{0mm}
	\left\{\hspace{0mm}
	\begin{aligned}
		&(\mathbf{P}, \mathbf{z}, \mathcal{M}), (\mathbf{x})\hspace{0mm}\in\hspace{0mm}\bigl(\mathbb{Z}_{q}^{m\times n}\times \mathbb{Z}_{q}^{m}\times ([1,n]^3)^\ell\bigr)\hspace{0mm}\times\hspace{0mm}(\mathbb{Z}_q^n) \hspace{0mm}:\\
		&\mathbf{P}\cdot\mathbf{x}=\mathbf{z}\wedge\forall(h,i,j)\in\mathcal{M},\mathbf{x}[h]=\mathbf{x}[i]\cdot\mathbf{x}[j]
	\end{aligned}
	\right\},
	\]		
	where $\mathbf{P}$, $\mathbf{z}$ and $\mathcal{M}$ are public, and $\mathbf{x}$ is the witness. 
	They jointly define the statement: the linear relation on $\mathbf{x}$ is specified by $\mathbf{P}$ and $\mathbf{z}$, and the quadratic relations on $\mathbf{x}$ are defined by the set of tuples $\mathcal{M}$.
	
	\begin{theorem}[\cite{YAZ+19_ZKP_on_Z}]
		If the ${LWE}_{n,q,\alpha}$ and ${SIS}_{n,q,\beta}$ assumptions hold, the scheme in \cite{YAZ+19_ZKP_on_Z} is a secure NIZKAoK with negligible completeness error and soundness error.
	\end{theorem}
	\subsubsection{Signature with Efficient Protocols}
	
	Signature with efficient protocols (SEP) is a specific type of signature introduced by Camenisch and Lysyanskaya \cite{CamenischL02}, which enables a user to prove possession of a valid signature on a message without revealing either the signature or the message itself. Our construction employs the signature (with efficient protocols) scheme proposed in \cite{JRS23_ACS}. This scheme is executed between a signer $S$ and a user $U$. The user $U$ interacts with $S$ to get an oblivious signature on a message, by only providing $S$ a commitment of the message. 
	
	Let $n$, $q$, $q'$, $m_1$, $m_2$ and $m_3$ be positive integers, where $q$ is a prime with $q\geq q'$. Set $m=m_1+m_2$ and $\mathcal{T}\xleftarrow{}\mathbb{Z}_{q'}\setminus\{0\}$. Select gaussian sampling width $\sigma$ and $\sigma_2$, and $\sigma_1=\sqrt{\sigma^2+\sigma_2^2}$. Sample a message commitment key $\mathbf{D}\xleftarrow{\$}\mathbb{Z}_q^{n\times m_3}$.
	\begin{enumerate}
		\item[-] $\left(\left(\mathbf{A},\mathbf{B},\mathbf{u}\right),\mathbf{R}\right)\leftarrow\mathbf{KeyGen}$$(1^\lambda)$. On input a security parameter $1^\lambda$, this algorithm samples matrices $\mathbf{A}\xleftarrow{\$}\mathbb{Z}_q^{n\times m_1}$ and $\mathbf{R}\xleftarrow{\$}[-1,1]^{m_1\times m_2}$, and a vector $\mathbf{u}\xleftarrow{\$}\mathbb{Z}_q^n$, then computes $\mathbf{B}=\mathbf{AR}\mod q\in\mathbb{Z}_q^{n\times m_2}$. Finally, the algorithm outputs a public key $\mathsf{pk}=(\mathbf{A},\mathbf{B},\mathbf{u})$ and a secret key $\mathsf{sk}=\mathbf{R}$.
		\item[-] $(\tau,\mathbf{v})\leftarrow\mathbf{Sign}$$(S(\mathbf{R},(\mathbf{A},\mathbf{B},\mathbf{u}))\leftrightarrow U(\mathbf{m}))$. This algorithm is an interaction protocol between $U$ and $S$. On input a message $\mathbf{m}\in\{0,1\}^{m_3}$, $U$ samples a vector $\mathbf{r}'\xleftarrow{\$}\mathcal{D}_{\mathbb{Z}^{m_1},\sigma_2}$, then computes a commitment to message, i.e. $\mathbf{c}=\mathbf{Ar'}+\mathbf{Dm}\mod q$, $U$ sends $\mathbf{c}$ together with a proof $\Pi_{m}$ that $\mathbf{c}$ is commitment to $\mathbf{m}$ with randomness $\mathbf{r}$' to $S$. On input the secret key $\mathsf{sk}=\mathbf{R}$, the public key $\mathsf{pk}=(\mathbf{A},\mathbf{B},\mathbf{u})$, and the commitment $\mathbf{c}$ and $\Pi_{m}$, $S$ verifies if $\Pi_{m}$ is valid, if not, $S$ aborts; Otherwise, $S$ samples a vector $\mathbf{r}''\xleftarrow{\$}\mathcal{D}_{\mathbb{Z}^{m_1},\sigma_2}$, computes $\mathbf{c}'=\mathbf{c}+\mathbf{Ar''}\mod q$, then samples a tag $\tau\xleftarrow{\$}\mathcal{T}$ and a vector $\mathbf{v}'\leftarrow\mathbf{SamplePre}(\mathbf{A}_\tau,\mathbf{R},\mathbf{u+c'},\sigma)-[\mathbf{r''}^{\top}|\mathbf{0}_{m_2}]^\top$, where $\mathbf{A}_\tau=[\mathbf{A}\:|\:\tau(\mathbf{I}_n\otimes \mathbf{g})-\mathbf{B}]\in \mathbb{Z}_q^{n\times m}$. Finally, $S$ sends $(\tau,\mathbf{v}')$ to $U$. After that, $U$ computes $\mathbf{v}=\mathbf{v}'-[\mathbf{r'}^\top|\mathbf{0}_{m_2}]$ and get the signature $(\tau,\mathbf{v})$. 
		
		\item[-] $(\mathsf{True}/\mathsf{False})\leftarrow\mathbf{Verify}((\mathbf{A},\mathbf{B},\mathbf{u}),\mathbf{m},(\tau,\mathbf{v}))$. On input the public key $\mathsf{pk}=(\mathbf{A},\mathbf{B},\mathbf{u})$, a message $\mathbf{m}\in\{0,1\}^{m_3}$ and a signature $\mathsf{sig}=(\tau,\mathbf{v})$, this algorithm computes the matrix $\mathbf{A}_\tau=[\mathbf{A}\:|\:\tau(\mathbf{I}_n\otimes \mathbf{g})-\mathbf{B}]$, also parses $\mathbf{v}$ as $[\mathbf{v}_1^\top|\mathbf{v}_2^\top]$ with $\mathbf{v}_1\in\mathbb{Z}^{m_1}$ and $\mathbf{v}_2\in\mathbb{Z}^{m_2}$. Finally, the algorithm outputs $\mathsf{True}$ if $\mathbf{A}_\tau\cdot[\mathbf{v}_1^\top|\mathbf{v}_2^\top]=\mathbf{u}+\mathbf{Dm}\mod q$ with $\|\mathbf{v}_1\|_\infty\leq\sigma_1\log{m_1}$ and $\|\mathbf{v}_2\|_\infty\leq\sigma_1\log{m_2}$; otherwise, it outputs $\mathsf{False}$ to indicate failure.
	\end{enumerate}
	
	\begin{theorem}[\cite{JRS23_ACS}]
		If the ${SIS}^{}_{n,q,\beta}$ assumption holds, the scheme in \cite{JRS23_ACS} is existentially unforgeable under chosen message attacks.
	\end{theorem}
	
	Furthermore, SEP is equipped with two efficient protocols: one to prove knowledge of an  opening of commitment and another to prove knowledge of a signature-message pair.
	For the above two protocols, a recent framework of ZKAoK from \cite{YAZ+19_ZKP_on_Z} is applied to instantiate them \cite{JRS23_ACS}. 
	
	\subsubsection{Accumulator}\label{dynamic_acc}
	
	In our construction, a static accumulator is used to ensure limited-times authentication and a dynamic accumulator is used to dynamically manage users. 
	
	We employ the static accumulator constructed by Libert et al. in \cite{LLNW16_static_acc}. 			
	This scheme proceeds as follows.
	
	Let $n$, $m$, $k$ and $q$ be positive integers, where $k=\lceil\log{q}\rceil$ and $m=2nk$.
	\begin{enumerate}
		\item[-]  $\mathbf{A}\leftarrow\mathbf{Setup}(1^\lambda).$ On input a security parameter $1^\lambda$, this algorithm picks a matrix $\mathbf{A}\xleftarrow{\$}\mathbb{Z}_q^{n\times m}$ randomly, and outputs $pp=\mathbf{A}$.
		\item [-] $\mathbf{u}\leftarrow\mathbf{Acc}(pp,D)$. On input a set $D$ consisting of the following $N$ binary vectors: $\mathbf{d}_0\in\{0,1\}^{nk},\ldots,\mathbf{d}_{N-1}\in\{0,1\}^{nk}$, this algorithm runs as follows.
		
		For any $j\in[0,N-1]$, $(j_1,\ldots,j_\ell)\in\{0,1\}^\ell$ denotes the binary decomposition of $j$, and $\mathbf{d}_j=\mathbf{u}_{j_1,\ldots,j_\ell}$. Form all nodes of tree with depth $\ell=\log N$ based on the $N$ leaves $\mathbf{u}_{0,0,...,0},\ldots,\mathbf{u}_{1,1,...,1}$ in the following way:
		$$\left\{
		\begin{array}{l@{\hspace{0.2cm}}l}
			\begin{aligned}
				\mathbf{u}_{b_1,...,b_i}&=h_\mathbf{A}(\mathbf{u}_{b_1,...,b_i,0},\mathbf{u}_{b_1,...,b_i,1})\in\{0,1\}^{nk},\\
				&\text{for }\forall (b_1,\ldots,b_i)\in\{0,1\}^i
			\end{aligned}; & \text{if }i\in[\ell]\\
			&\\
			\mathbf{u}= h_\mathbf{A}(\mathbf{u}_0,\mathbf{u}_1)\in\{0,1\}^{nk}; & \text{if }i= 0
		\end{array}
		\right.,$$
		where $i$ represents the depth of tree. 
		The algorithm outputs the accumulator value $\mathbf{u}$.
		\item[-] $\mathbf{w}\leftarrow\mathbf{Wit}(pp,D,\mathbf{d}).$ On input a set $D$ and a vector $\mathbf{d}$, this algorithm runs as follows. 
		
		If $\mathbf{d}\not\in D$, the algorithm outputs $\bot$. Otherwise, $\mathbf{d}=\mathbf{d}_j\in D$ for some $j\in[0,N-1]$
		with binary decomposition $(j_1,\ldots,j_\ell)$, the algorithm outputs the witness $\mathbf{w}$ defined as:
		$$\begin{aligned}
			\mathbf{w}=((j_1,\ldots,j_\ell),&(\mathbf{u}_{j_1,\ldots,j_{\ell-1},\bar{j}_\ell},\ldots,\mathbf{u}_{j_1,\bar{j}_2},\mathbf{u}_{\bar{j}_1}))\\
			&\in\{0,1\}^\ell\times\left(\{0,1\}^{nk}\right)^\ell,
		\end{aligned}$$ where $\mathbf{u}_{j_1,...,j_{\ell-1},\bar{j_\ell}},\ldots,\mathbf{u}_{j_1,\bar{j_2}},\mathbf{u}_{\bar{j_1}}$ are calculated as specified in the algorithm $\mathbf{Acc}_\mathbf{}(pp,D)$.
		\item[-] $(\mathsf{True}/\mathsf{False})\leftarrow\mathbf{Verify}(pp,\mathbf{u},\mathbf{d},\mathbf{w})$. On input a accumulator value $\mathbf{u}$, a vector $\mathbf{d}$, and a witness $\mathbf{w}$ parsed as:
		$$\mathbf{w}=((j_1,\ldots,j_\ell),(\mathbf{w}_\ell,\ldots,\mathbf{w}_1))\hspace{0mm}\in\hspace{0mm}\{0,1\}^\ell\times(\{0,1\}^{nk})^\ell,$$
		from bottom to top, this algorithm generates the auxiliary node $\mathbf{v}_\ell,\mathbf{v}_{\ell-1},\ldots,\mathbf{v}_1,\mathbf{v}_0\in\{0,1\}^{nk}$ recursively, which is as follows: $\mathbf{v}_\ell=\mathbf{d}$, and
		$$\mathbf{v}_{i}=\begin{cases}h_{\mathbf{A}}(\mathbf{v}_{i+1},\mathbf{w}_{i+1}),\:\mathrm{if}\:j_{i+1}=0;\\h_{\mathbf{A}}(\mathbf{w}_{i+1},\mathbf{v}_{i+1}),\:\mathrm{if}\:j_{i+1}=1,\end{cases}$$
		for $\forall i\in\{\ell-1,\ldots,1,0\}$.
		Then it returns $\mathsf{True}$ if $\mathbf{v}_0 = \mathbf{u}$. Otherwise, it returns $\mathsf{False}$.
	\end{enumerate}
	
	\begin{theorem}[\cite{LLNW16_static_acc}]
		If the ${SIS}^{\infty}_{n,q,1}$ assumption holds, the scheme in \cite{LLNW16_static_acc} is a secure accumulator with collision-resistance property.
	\end{theorem}
	Our construction applies the dynamic accumulator introduced by Zhao et al. \cite{ZYH24_dynamic_acc} to implement dynamic grant and revocation. 
	This scheme works as follows.
	
	Let $n$, $m$, $k$, $q$, $N$ and $l$ be positive integers, where $N\in$ $\mathbb{N}>2$, $l=\lceil\log_{}{N}\rceil$, $q>N$, $k=\lceil\log_{}{q}\rceil$, and $m=\mathcal{O}(n\log q)$.
	\begin{enumerate}
		\item[-] $\left((\mathbf{U},\mathbf{A}), \mathbf{R}\right)\leftarrow\mathbf{KeyGen}(1^\lambda){.}$ 
		On input a security parameter $1^\lambda$, this algorithm first runs $(\mathbf{A},\mathbf{T})\leftarrow$ \textbf{TrapGen}$(1^n,1^m,q)$. Second, it selects a random matrix $\mathbf{U}=[\mathbf{U}_0\mid\cdots\mid\mathbf{U}_{N-1}]\stackrel{s}{\leftarrow}\mathbb{Z}_{q}^{n\times lN}$, where $\mathbf{U}_i\in\mathbb{Z}_q^{n\times l}$. Let $\mathfrak{s}_1=\mathfrak{s}_1(\mathbf{T})$, and the algorithm generates a Gaussian parameter $s$ satisfying $s \geq \mathfrak{s}_{1}^{2}\cdot\omega(\sqrt{\log n})$. For every $i\in[N]$, it samples $\mathbf{R}_{i}\leftarrow$ \textbf{SamplePre}$(\mathbf{A},\mathbf{T},s,\mathbf{U}_{i})$, such that $\mathbf{A}\mathbf{R}_i=\mathbf{U}_i$. Finally, the algorithm outputs $pk=(\mathbf{U},\mathbf{A})$, $sk=\mathbf{R}=(\mathbf{R}_i)_{i\in[N]}\in\mathbb{Z}^{m\times lN}$.
		\item[-] $(\mathbf{u},\mathbf{r})\leftarrow\textbf{Acc}$$(pk,X\subseteq[N]){.}$ On input the public key $pk$ and a set $X\subseteq[N]$, this algorithm parses $\mathbf{m}=[\mathbf{m}_0^{\top}\mid\cdots\mid\mathbf{m}_{N-1}^{\top}]^{\top}\in\{0,1\}^{lN}$, and for $i\in[N]$, if $i\in X$, $\mathbf{m}_i$ is a $l$-dimensional all-ones vector, i.e. $\mathbf{m}_i=1^l$; otherwise, $\mathbf{m}_i$ is a $l$-dimensional all-zeros vector, i.e. $\mathbf{m}_i=0^l$. Then the algorithm chooses a vector $\mathbf{r}\leftarrow D_{\mathbb{Z}^{m},s}$, and generates an accumulator value via $\mathbf{u}=\mathbf{Um}+\mathbf{Ar}\in\mathbb{Z}_q^{n}$. Finally, it outputs the accumulator value $\mathbf{u}$ and a auxiliary vector $\mathbf{r}$.
		\item[-] $\mathbf{w}\leftarrow\textbf{WitGen}$$(sk,X,\mathbf{r},i\in[N]){.}$ 
		On input the secret key $sk$, an accumulator value $\mathbf{u}$, an element $i\in[N]$, and a vector $\mathbf{r}$, this algorithm computes a witness $\mathbf{w}=\sum_{j\in[N]\setminus\{i\}}\mathbf{R}_{j}\mathbf{m}_{j}+\mathbf{r}\in\mathbb{Z}_{q}^{m}$, where $\mathbf{m}_{j}$ is computed according to $X$. Finally, it outputs the witness $\mathbf{w}$.
		\item[-] $(\mathsf{True}/\mathsf{False})\leftarrow\textbf{Verify}$$(pk,\mathbf{u},j,\mathbf{w}){.}$ On input the public key $pk$, an accumulator value $\mathbf{u}$, an element $j$, and a witness $\mathbf{w}$, this algorithm outputs $\mathsf{True}$ if $\|\mathbf{w}\|_\infty\leq\gamma$ (or $\|\mathbf{w}\|_2\leq\gamma_1$) and $\mathbf{u}=\mathbf{Aw}+\mathbf{U}\mathbf{y}$; otherwise, it outputs $\mathsf{False}$, where $\gamma=O(s\sqrt{lN})<q$, $\gamma_{1}=O(s(m+lN))<q$, and $\mathbf{y}=\vartheta_{j}^{N}$.
		\item[-] $\left(\mathbf{u}',\text{state}\right)\leftarrow\textbf{UpdateAcc}$$(pk,sk,j,\mathbf{u},X){.}$ On input the public key $pk$, the secret key $sk$, an element $j$, an accumulator value $\mathbf{u}$, and a set $X$, this algorithm sets $X=X\cup\{j\}$, $\mathbf{m}'_j=1^l$ and $\mathbf{m}^{*} = \mathbf{m}'_j - \mathbf{m}_{j}$ when $j\notin X$ needs to be added in $X$. And it sets $X=X\setminus\{j\}$, $\mathbf{m}_j^{\prime}=0^{l}$ and $\mathbf{m}^{*}=\mathbf{m}_j^{\prime}-\mathbf{m}_j$ when $j\in X$ needs to be deleted from $X$. Now, the algorithm generates an updated accumulator value via $\mathbf{u}'=\mathbf{u}+\mathbf{U}_j\mathbf{m}^*\in\mathbb{Z}_q^n$, then computes $\mathbf{t}_j=\mathbf{R}_j\mathbf{m}^*\in\mathbb{Z}^m$. Finally, the algorithm outputs the new accumulator value $\mathbf{u}'$, and a state information $\mathsf{state}=((j,\mathbf{m}^{*}),\mathbf{t}_j)$.
		\item[-] $\mathbf{w}'\leftarrow\textbf{UpdateWit}$$(i,\mathbf{w},\mathsf{state}){.}$ On input an element $i$, a witness $\mathbf{w}$, and a state information, this algorithm computes an updated witness $\mathbf{w}'=\mathbf{w}+\mathbf{t}_j\in\mathbb{Z}^m$. Finally, it outputs the new witness $\mathbf{w}'$ for user $i$.
	\end{enumerate}
	
	\begin{theorem}[\cite{ZYH24_dynamic_acc}]
		If the ${SIS}^{}_{n,q,\beta}$ assumption holds, the scheme in \cite{ZYH24_dynamic_acc} is a secure dynamic accumulator with collision resistance property.
	\end{theorem}
	
	\section{Main Construction}
	
	\subsection{Dynamic $k$-TAA Scheme}\label{Dynamic_k-TAA}
	
	In this section, we propose a lattice-based dynamic $k$-TAA scheme with attribute-based credentials.
	Our scheme is derived from the wPRF with efficient protocols (with strong uniqueness) $\mathsf{PRF}=(\mathbf{KeyGen}$, $\mathbf{Eval})$ \cite{Rupeng_wPRF}, the accumulator scheme $\mathsf{ACC}=(\mathbf{Setup}$, $\mathbf{Acc}$, $\mathbf{Wit}$, $\mathbf{Verify})$ \cite{LLNW16_static_acc}, the dynamic accumulator scheme $\mathsf{DACC}=(\mathbf{KeyGen}$, $\mathbf{Acc}$, $\mathbf{WitGen}$, $\mathbf{Verify}$, $\mathbf{UpdateAcc}$, $\mathbf{UpdateWit})$ \cite{ZYH24_dynamic_acc}, the signature scheme $\mathsf{SIGN}=(\mathbf{KeyGen}$, $\mathbf{Sign}$, $\mathbf{Verify})$ \cite{JRS23_ACS}, a ZKAoK framework \cite{YAZ+19_ZKP_on_Z}, and a hash function defined as $\mathsf{H}:\{0,1\}^{*}\to\mathbb{Z}_{q_1}^{m_D\times n}\times\mathbb{Z}_{q_1}^{m_D\times n}$, it works as follows.
	\begin{enumerate}
		\item[-]  
		\textbf{Public Parameters.} We choose positive integers $q_0$, $p$, $q$, $q_1$ and $q'$. Let $e_1, e_2, e_3$ be positive integers that $1\leq e_1<e_2<e_3$,  and $q_0$ be a prime number, we set $p=q_0^{e_1}$, $q_1=q_0^{e_2}$, $q=q_0^{e_3}$, where $q_1$ is modulus for wPRF and $q$ is modulus for other components, and $q'$ is modulus for tag space of signature ($q'<q_0$). Let $k_{1}=\lceil\log p\rceil$, $k_{2}=\lceil\log q_1\rceil$ and $k=\lceil\log{q}\rceil$. Then we choose integers $n\geq 2$, $n_1\geq 2$, $m=\mathcal{O}(n\log{q})$, and $m^{\prime}=n_{}k$, where $n$ is the SIS dimension for signature and dynamic accumulator, $n_1$ is the SIS dimension for other components.
		Let $N(=|\mathbb{Z}_{q^{\prime}}|)\in\mathbb{N}$ with $N>2$, and $l=\lceil \log_{2}{N}\rceil$ (we require $q>N$). 
		Set $m_A=\mathcal{O}(n\log{q})$,
		$m_S=\mathcal{O}(n_1\log{q})$, 
		and $m_D \geq 2n\cdot (\log q_1+1)/(\log p-1)$. 
		Select a positive integer $q$ with $q^{\prime}\leq q$, and set the tag space $\mathcal{T}\leftarrow\mathbb{Z}_{q^{\prime}}\setminus\{0\}$. Suppose that  $m_1=m$, and $m_2= m^{\prime}$. 
		Set $m_e$ as the message length of signature,  $\sigma=\mathcal{O}\left(\sqrt{m_1}+\sqrt{m_2}\right)$ (Gaussian sampling width for signatures), $\sigma_2=\max\mathcal{O}\left(\left(\sqrt{m_1}+\sqrt{m_2}\right),\omega\left(\log{m_1}\right)\right)$ (Gaussian sampling width for commitment randomness), and $\sigma_1=\sqrt{\sigma^2+\sigma_2^2}$. 
		In addition, the modulus employed in the ZKAoK of \cite{YAZ+19_ZKP_on_Z} is $q$.
		Finally, $\mathbf{E}$ is employed to map the sum of two wPRF inputs into the value space of the static accumulator, while $\mathbf{F}$ is used to map the wPRF's output into the message space of the signature scheme, and matrix $\mathbf{A}$ serves as a public input for the wPRF.
		\\
		\item[-] 
		\textbf{UserSetup.} 
		A user $i$ runs $\mathsf{PRF}.\mathbf{KeyGen}(1^\lambda)$ to generate his/her secrete key  $usk=\mathbf{k}$, then computes $\mathsf{PRF}.\mathbf{Eval}(\mathbf{k},\mathbf{A}_{})$ to get his/her public key $upk=\mathbf{y}$. 
		\\
		\item[-] 
		\textbf{GMSetup.}
		The GM runs $\mathsf{SIGN}.\mathbf{KeyGen}(1^\lambda)$ to generate a public key $gpk=(\mathbf{A}_{GM},\mathbf{B}_{GM},\mathbf{u})$ and a secret key $gsk=\mathbf{R}_{GM}$. 
		The GM also maintains an identification list, denoted as $\mathcal{LIST}$, which consists of entries of the form $(i, \mathbf{y}, \tau)$, where $i$ denotes user $i$, $\mathbf{y}$ is the identification representation of the user, and $\tau$ is referred to as the membership public key of the user. \\
		\item[-] 
		\textbf{APSetup.} Each AP publishes the tuple  $(ID_{AP}$, $k_{AP})$, where $ID_{AP}$ denotes the AP's identity and $k_{AP}$ represents the allowed number of times that a group member can access the AP's application.
		Thus, each entity in the system is able to generate the tag bases set $\mathcal{B}$ corresponding to the AP, where each element is a pair $(\mathbf{B}_j, \check{\mathbf{B}}_j) = \mathsf{H}(ID_{AP}, k_{AP}, j)$ for $j \in [1, k_{AP}]$.
		In addition, each AP executes $\mathsf{DACC}.\mathbf{KeyGen}(1^\lambda)$ to obtain
		$apk=(\mathbf{U},\mathbf{A}_{AP})$ and $ask=\mathbf{R}_{AP}$. The AP maintains an authentication log $\mathcal{LOG}$, and also a public archive $ARC$ consisting of entries of the form $(\mathsf{state\_info}, \mathsf{acc\_value})$, where $\mathsf{state\_info}$ is a part of state information generated by AP via $\mathsf{DACC}.\mathbf{UpdateAcc}((\mathbf{U},\mathbf{A}_{AP}),\mathbf{R}_{AP},j,\mathbf{u},X)$, and $\mathsf{acc\_value}$ is the updated accumulator value after the member is granted or revoked by AP. The $\mathcal{LOG}$ and $ARC$ are initialized as empty sets, while AP runs  $\mathsf{DACC}.\mathbf{Acc}((\mathbf{U},\mathbf{A}_{AP}),X=\emptyset\subseteq[N])$ to get the accumulator value $\mathbf{u}_0$ and a random vector $\mathbf{r}$.\\
		\item[-] 
		\textbf{Join.} User with a set of attributions $\mathsf{attr}:=\{\mathsf{attr}_i\}_{i\in[m_4]}$ sends his/her public key $\mathbf{y}$ to the GM, and proves the knowledge of the $usk$ $\mathbf{k}$ corresponding to $upk$ $\mathbf{y}$ by $\Pi_U=\text{SPK}\{(\mathbf{k}):\mathbf{y}=\mathsf{PRF}.\mathbf{Eval}(\mathbf{k},\mathbf{A})\}[\mathfrak{m}]$. If $\mathbf{y}$ has not yet been recorded in the $\mathcal{LIST}$, the GM first appends entry $(i,\mathbf{y})$ to the $\mathcal{LIST}$, then user computes $\mathbf{m}=\mathsf{bin}(\mathbf{F}\cdot \mathsf{bin}(\mathbf{y}))$, and runs $\mathsf{SIGN}.\mathbf{Sign}(S((\mathbf{A}_{GM},\mathbf{B}_{GM},\mathbf{u}),\mathbf{R}_{GM}),U(\mathbf{m},\mathsf{attr}))$ with GM to get a signature on $\tilde{\mathbf{m}}:=(\mathbf{m},\mathsf{attr})$. Then, it returns the signature $(\tau, \mathbf{v})$ back to user $i$, and updates the corresponding entry $(i, \mathbf{y})$ by including $\tau$, resulting in $(i, \mathbf{y}, \tau)$.
		At this point, 
		GM maintains an identification list containing $(i,\mathbf{y},\tau)$. So user's $mpk$ is $\tau$, and $msk$ is $(\mathbf{v},(\mathbf{m},\mathsf{attr}))$.
		\\
		\item[-] 
		\textbf{GrantingAccess.} The AP grants access to a specific user with the membership public key $\tau$, more concretely, $\tau$ is included into the latest accumulator value. Let us assume that the AP's $ARC$ contains $j$ tuples, and the current accumulator value being $\mathbf{u}_j$.
		The AP calculates a new accumulator value and state information $(\mathbf{u}_{j+1},((j,\mathbf{m}^{*}),\mathbf{t}_j))$ via $\mathsf{DACC}.\mathbf{UpdateAcc}((\mathbf{U},\mathbf{A}_{AP}),\mathbf{R}_{AP}, \tau, \mathbf{u}_j, X)$, and runs $\mathsf{DACC}.\mathbf{WitGen}((\mathbf{U},\mathbf{A}_{AP}),\mathbf{R}_{AP},X,\mathbf{r},\tau\in[N])$ to generate a witness $\mathbf{w}$ for the user. Then the AP appends $(\mathbf{t}_j, \mathbf{u}_{j+1})$ to the $ARC$, and sends $\mathbf{w}$ to the user.
		Existing users in the access group of AP update their witnesses using the $\mathsf{state\_info}$ information from $ARC$ as follows. User $k$ with witness $\mathbf{w}$ runs $\mathsf{DACC}.\mathbf{UpdateWit}(k,\mathbf{w},ARC)$ to calculate an updated witness $\mathbf{w}^*=\mathbf{w}+\mathbf{t}_j\in\mathbb{Z}^m$ for himself/herself.\\
		\item[-] 
		\textbf{RevokingAccess.}  The AP revokes access of a specific user with the membership public key $\tau$, more concretely, $\tau$ is excluded from the latest accumulator value. Assume that the AP's $ARC$ contains $j$ tuples, and the current accumulator value being $\mathbf{u}_j$.
		The AP computes a new accumulator value $\mathbf{u}_{j+1}$ and state information $((j,\mathbf{m}^{*}),\mathbf{t}_j)$ via $\mathsf{DACC}.\mathbf{UpdateAcc}((\mathbf{U},\mathbf{A}_{AP}),\mathbf{R}_{AP}, \tau, \mathbf{u}_j, X)$. It then appends $(\mathbf{t}_j, \mathbf{u}_{j+1})$ to $ARC$. In a manner similar to \textbf{GrantingAccess}, existing users in the access group of AP update their witnesses using the $\mathsf{state\_info}$ from the $ARC$, user $k$ with witness $\mathbf{w}$ get an updated witness $\mathbf{w}^*=\mathbf{w}+\mathbf{t}_j\in\mathbb{Z}^m$.\\
		\item[-] 
		\textbf{Authentication.} The user maintains a counter list, with each entry corresponding to an AP, to ensure that the user does not authenticate more than $k_{AP}$ times with any given AP. User with public key $\mathbf{y}$ and attributes vector $\mathsf{attr}$ authenticates himself/herself by the following interactive protocol. Users can selectively disclose their attributes according to the requirements of the AP. For example, if the AP requires users to possess the attributes $\{\mathsf{attr}_i\}_{i \in \mathcal{I}}$ with $\mathcal{I} \subseteq [m_4]$, then the user only discloses these required attributes to convince the AP, while the remaining attributes $\{\mathsf{attr}_i\}_{i \notin \mathcal{I}}$ stay hidden.
		
		Before interacting with an AP, the user first verifies the number of authentications associated with that AP has not exceeded $k_{AP}$, if not, then calculates the tag bases set $\mathcal{B}$ which consists of $(\mathbf{B}_{j_{}}, \check{\mathbf{B}}_{j_{}}) = \mathsf{H}(ID_{AP}, k_{AP}, j_{})$ for $j_{}\in[1,k_{AP}]$, and constructs a set $\mathcal{B}^{\prime}$, where each element is defined as $\mathbf{b}_j^{\prime} = \mathsf{bin}(\mathbf{E} \cdot \mathsf{vdec}(\mathsf{M2V}(\mathbf{B}_j + \check{\mathbf{B}}_j)))$ for every $(\mathbf{B}_j, \check{\mathbf{B}}_j) \in \mathcal{B}$.
		It is assumed that the user is conducting $j$-th access to the AP's application, the user calculates $(\mathbf{B}, \check{\mathbf{B}}) = \mathsf{H}(ID_{AP}, k_{AP}, j_{AP})$, then $\mathbf{b}^{\prime} = \mathsf{bin}(\mathbf{E} \cdot \mathsf{vdec}(\mathsf{M2V}(\mathbf{B} + \check{\mathbf{B}})))$.
		With $\mathcal{B}^{\prime}$ and $\mathbf{b}^\prime$ in place, the user proceeds to calculate the accumulator value $\mathbf{u}^{\prime} = \mathsf{ACC}.\mathbf{Acc}(\mathcal{B}^{\prime})$ together with its witness $\mathbf{w}^{\prime} = \mathsf{ACC}.\mathbf{Wit}(\mathcal{B}^{\prime}, \mathbf{b}^{\prime})$.
		After preparing these, the user sends a request to the AP then receives a challenge $(c, \mathfrak{m})$, where $c$ is a challenge uniformly chosen from a challenge space (e.g. $\mathbb{Z}_p$ for a prime $p$) and $\mathfrak{m}$ is an arbitrary message. Finally, the user calculates the following response information:
		$$\mathbf{t} = \mathsf{PRF}.\mathbf{Eval}(\mathbf{k}, \mathbf{B}),$$ 
		and
		$$\check{\mathbf{t}} = \mathsf{PRF}.\mathbf{Eval}(\mathbf{k}, \check{\mathbf{B}}) + c \cdot \mathbf{y} \bmod p,$$ 
		then generates the proof $\Pi =$
		
		{\small\[ \hspace{-0mm}\text{SPK}\hspace{-0mm}\left\{
			\begin{aligned}\label{Pi}
				&(\mathbf{k},\mathbf{y},\mathbf{m},(\tau, \mathbf{v}),\mathbf{w},(\mathbf{B},\check{\mathbf{B}}),\mathbf{b}^{\prime},\check{\mathbf{t}}^{\prime},\mathbf{w}^{\prime},\{\mathsf{attr_i}\}_{i\notin\mathcal{I}}):\\
				& \mathbf{y}=\mathsf{PRF}.\mathbf{Eval}(\mathbf{k},\mathbf{A}), \\ 
				& \mathbf{G}_{n,p-1}\cdot\mathbf{m}=\mathbf{F}\cdot \mathsf{bin}(\mathbf{y}) \mod p,\\
				& \mathsf{DACC}.\mathbf{Verify}((\mathbf{U},\mathbf{A}_{AP}),\mathbf{u}_{AP},\tau,\mathbf{w})=1, \\
				& \mathsf{SIGN}.\mathbf{Verify}((\mathbf{A}_{GM},\mathbf{B}_{GM},\mathbf{u}), (\mathbf{m},\mathsf{attr}), (\tau, \mathbf{v}))=1,\\
				& \mathbf{t}=\mathsf{PRF}.\mathbf{Eval}(\mathbf{k},\mathbf{B}),\\
				& \check{\mathbf{t}}^{\prime}=\mathsf{PRF}.\mathbf{Eval}(\mathbf{k},\check{\mathbf{B}}), \\
				& \check{\mathbf{t}}^{\prime}+ c\cdot \mathbf{y}= \check{\mathbf{t}} \mod p, \\%
				& \mathbf{G}_{n,q_{1}-1}\cdot \mathbf{b}^{\prime}=\mathbf{E}_{}\cdot \mathsf{vdec}(\mathsf{M2V}(\mathbf{B}+\check{\mathbf{B}})) \text{ mod } q_1,\\
				&  \mathsf{ACC}.\mathbf{Verify}_{}(\mathbf{u}^{\prime},\mathbf{b}^{\prime},\mathbf{w}^{\prime})=1.
			\end{aligned}
			\right\}\hspace{-0mm}[\mathfrak{m}].
			\]}

		Subsequently, the user responds to AP with tuple $(\mathbf{t}, \check{\mathbf{t}}, \Pi)$. Notably, all the terms mentioned above can be proved using ZKAoK framework of \cite{YAZ+19_ZKP_on_Z}. 
		We have provided the detailed proofs in \ref{secA1}. 
		When AP receives the tuple $(\mathbf{t}, \check{\mathbf{t}}, \Pi)$, it first calculates the accumulator value $\mathbf{u}^{\prime}$ as the user did. Then, it checks the validity of $\Pi$ and appends the transcript $(\mathbf{t}, \check{\mathbf{t}}, c, \mathfrak{m}, \Pi)$ to its $\mathcal{LOG}$. The AP accepts the user if $\mathbf{t}$ is not already in $\mathcal{LOG}$ and $\Pi$ is valid.\\
		\item[-] 
		\textbf{PublicTracing}. This algorithm takes as input the $\mathcal{LOG}$: $\left\{\left(\mathbf{t}_i,\check{\mathbf{t}}_i,c_i,\mathfrak{m}_i,\Pi_i\right)\right\}_{i=1}^\mathit{\Gamma}$ that consists of $\mathit{\Gamma}$ elements and the $\mathcal{LIST}$, then the algorithm initializes an empty set $\mathcal{S}$. 
		For each pair $m,n \in [1,\mathit{\Gamma}]$ such that $\mathbf{t}_m = \mathbf{t}_n$ and $c_m \neq c_n$, with both $\Pi_m$ and $\Pi_n$ being valid, the algorithm computes $\mathbf{y} = (c_m - c_n)^{-1} \cdot (\check{\mathbf{t}}_m - \check{\mathbf{t}}_n)$. Since $c_m - c_n$ is invertible due to the proper challenge space selection, the algorithm searches for $\mathbf{y}$ in the $\mathcal{LIST}$. If it finds an entry containing $\mathbf{y}$, it appends $i$ to $\mathcal{S}$; otherwise, it appends GM to $\mathcal{S}$. Finally, this algorithm returns the set $\mathcal{S}$.
	\end{enumerate}

	\subsection{Security Analysis}
	Our dynamic $k$-TAA scheme is proven secure by Theorem \ref{ktaa_security}, as presented below.
	
	\begin{theorem}\label{ktaa_security}
		If wPRF with efficient protocols (with strong uniqueness) $\mathsf{PRF}$, static accumulator scheme $\mathsf{ACC}$, dynamic accumulator scheme $\mathsf{DACC}$, signature with efficient protocols $\mathsf{SIGN}$, and the  ZKAoK are all secure, the SIS assumption holds, and $\mathsf{H}$ is regarded as a random oracle, our dynamic $k$-times anonymous authentication scheme is secure.
	\end{theorem}
	\begin{proof}
		Our dynamic \(k\)-TAA scheme satisfies the following security properties: (1) D-Anonymity, (2) D-Detectability, (3) D-Exculpability for users, and (4) Exculpability for the group manager \cite{dynamic_k-taa_NS05}.
		We now proceed to prove that these properties, and
		in the following, we use $\mathcal{G}_i$ to denote Game i, $\mathcal{A}$ to denotes the adversary, $\mathcal{C}$ to denotes the challenger, and $\Pr_i$ to denote the probability that $\mathcal{A}$ succeeds in $\mathcal{G}_i$.
		\\\\
		\textbf{\textsl{D-Anonymity.}}
		The proof of D-Anonymity for our dynamic $k$-TAA scheme is based on the games defined as follows.
		\begin{enumerate}
			\item[{$\mathcal{G}_0$}]  is the real game of \textbf{D-Anonymity} defined in \ref{Formal_Definitions_Requirements}.
			\item[{$\mathcal{G}_1$}]  is same with $\mathcal{G}_0$, with the difference that the challenger $\mathcal{C}$ simulates a proof $\Pi$ when it needs to perform authentication representing $i_0$ or $i_1$. The zero-knowledge property of the ZKAoK implies that $| \Pr_1-\Pr_0|\leq negl(\lambda)$.
			\item[{$\mathcal{G}_2$}]  is same with $\mathcal{G}_1$, with the difference that $i_0$'s $upk$ $\mathbf{y}_{i_0}$ is sampled uniformly from the output space of $\mathsf{PRF}.\mathbf{Eval}$, and similarly, $\mathcal{C}$ selects a response $(\mathbf{t},\check{\mathbf{t}})$ uniformly at random from the output space of $\mathsf{PRF}.\mathbf{Eval}$ when it needs to perform authentication representing $i_0$. The weak pseudorandomness property of the $\mathsf{PRF}$ implies that $|\Pr_{2}-\Pr_{1}|\leq negl(\lambda)$. 
			\item[{$\mathcal{G}_3$}] is same with $\mathcal{G}_2$, with the difference that $i_1$'s $upk$ $\mathbf{y}_{i_1}$ is sampled uniformly from the output space of $\mathsf{PRF}.\mathbf{Eval}$, and similarly, $\mathcal{C}$ selects a response $(\mathbf{t},\check{\mathbf{t}})$ uniformly at random from the output space of $\mathsf{PRF}.\mathbf{Eval}$ when it needs to perform authentication representing $i_1$. The weak pseudorandomness property of the $\mathsf{PRF}$ implies that $|\Pr_{3}-\Pr_{2}|\leq negl(\lambda)$. 
			
		\end{enumerate}
		
		In $\mathcal{G}_3$, $\mathcal{A}$ is allowed to access oracle $\mathcal{O}_{Query}$ once of the form $(d,*)$ for $d\in\{0,1\}$, and is also allowed to access oracle $\mathcal{O}_{Auth-U}[gpk]$ at most $k-1$ times of the form $(i,(ID^*,M^*))$ for $i\in\{i_0,i_1\}$, hence $\mathcal{C}$ only need to perform authentication with AP associated with $(ID^*,k^*)$ at most $k^*$ times representing $i_0$ and $i_1$, respectively. Under these conditions, the query oracle behaves normally for $d \in \{0,1\}$, meaning that its executions are independent of the challenge bit $b$. Therefore, $\mathcal{A}$’s view in $\mathcal{G}_3$ does not depend on $b$, and its advantage in this game is $\Pr_3 = 1/2$. So, we obtain $|\Pr_0-1/2|\leq negl(\lambda)$. This concludes the proof for D-Anonymity of our scheme.\\
		
		Before proving D-Detectability of the dynamic $k$-TAA scheme in detail, we first provide the following lemma.
		\begin{lemma}
			Let $\mathbf{y}_i\in\mathbb{Z}_{p}^m$, $\mathbf{m}_i\in\{0,1\}^{m_3}$, $\mathbf{F}\in\mathbb{Z}_p^{n'\times m'}$, $\mathsf{F}:\mathbf{m}_i=\mathsf{bin}(\mathbf{F}\cdot \mathsf{bin}(\mathbf{y}_i)\mod p)$ is the map from $\mathbb{Z}_{p}^m$ to $\{0,1\}^{m_3}$. Then $\mathsf{F}$ is a collision-resistant function. Namely, $\mathrm{Pr}[\mathbf{y}_1\neq\mathbf{y}_2:\mathsf{bin}(\mathbf{F}\cdot \mathsf{bin}(\mathbf{y}_1)\mod p)=\mathsf{bin}(\mathbf{F}\cdot \mathsf{bin}(\mathbf{y}_2)\mod p)]\leq negl(\lambda)$.
		\end{lemma}
		\begin{proof}
			Suppose for $\mathbf{m}_i=\mathsf{bin}(\mathbf{F}\cdot \mathsf{bin}(\mathbf{y}_i)\mod p)$, or equivalently $\mathbf{G}_{n,q_1-1}\cdot\mathbf{m}_i=\mathbf{F}\cdot \mathsf{bin}(\mathbf{y}_i) \mod p$, we have $\mathbf{m}_1=\mathbf{m}_2$, but $\mathbf{y}_1\neq \mathbf{y}_2$. In other words, there exists a collision. So $\mathbf{F}\cdot \mathsf{bin}(\mathbf{y}_1)=\mathbf{F}\cdot \mathsf{bin}(\mathbf{y}_2)$. Then $\mathsf{bin}(\mathbf{y}_1)-\mathsf{bin}(\mathbf{y}_2)$ is a short non-zero vector, so it is a $\text{SIS}^\infty_{n',m',p,1}$ solution, where $n^\prime=\lfloor\frac{m_3}{\log{p}}\rfloor$ and $m^\prime=\lfloor m\cdot\log{p}\rfloor$.
		\end{proof}
		\noindent\textbf{\textsl{D-Detectability.}} The proof of D-Detectability for our dynamic $k$-TAA scheme is based on the games defined as follows.
		\begin{enumerate}
			\item[{$\mathcal{G}_0$}] is the real game of \textbf{D-Detectability} defined in \ref{Formal_Definitions_Requirements}.
			\item[{$\mathcal{G}_1$}] 
			is same with $\mathcal{G}_0$, with the difference that during the verification of $\mathcal{A}$’s success,the challenger  $\mathcal{C}$ also tries to extract a witness $(\mathbf{k},\mathbf{y},\mathbf{m},(\tau, \mathbf{v}),\mathbf{w},(\mathbf{B},\check{\mathbf{B}}),\mathbf{b}^{\prime},\check{\mathbf{t}}^{\prime},\mathbf{w}^{\prime},\{\mathsf{attr_i}\}_{i\notin\mathcal{I}})$ from every proof $\Pi$ in the $\mathcal{LOG}_{AP}$. If $\mathcal{C}$ is unable to extract a valid witness\footnote{Here, a witness is considered valid if it satisfies the statements which are defined in the authentication algorithm.} from any proof $\Pi$, it outputs 0, which indicates $\mathcal{A}$ fails. The soundness property of the ZKAoK systems implies that $|\Pr_1 - \Pr_0| \leq negl(\lambda)$.
			\item[{$\mathcal{G}_2$}]  is same with $\mathcal{G}_1$, with the difference that during the verification of $\mathcal{A}$’s success, if for an extracted witness $(\mathbf{k},\mathbf{y},\mathbf{m},(\tau, \mathbf{v}),\mathbf{w},(\mathbf{B},\check{\mathbf{B}}),\mathbf{b}^{\prime},\check{\mathbf{t}}^{\prime},\mathbf{w}^{\prime},\{\mathsf{attr_i}\}_{i\notin\mathcal{I}})$, \textcolor{black}{$\mathbf{m}$} has not been queried to the $\mathcal{O}_{Join-GM}[gpk,gsk]$ oracle and \textcolor{black}{$\tau$ has not been queried to $\mathcal{O}_{GRAN-AP}$ for some AP}, $\mathcal{C}$ outputs 0. The unforgeability of the signature scheme and security of the dynamic accumulator imply that $|\Pr_2-\Pr_1|\leq negl(\lambda)$.
			\item[{$\mathcal{G}_3$}] is same with $\mathcal{G}_2$, with the difference that during the verification of $\mathcal{A}$’s success, if for two extracted witness $(\mathbf{k}_1,\mathbf{y}_1,\mathbf{m}_1,(\tau_1, \mathbf{v}_1),\mathbf{w}_1,(\mathbf{B}_1,\check{\mathbf{B}}_1),\mathbf{b}^{\prime}_1,\check{\mathbf{t}}^{\prime}_1,\mathbf{w}^{\prime}_1,\{\mathsf{attr_i}\}_{i\notin\mathcal{I}}^1)$
			and $(\mathbf{k}_2,\mathbf{y}_2,\mathbf{m}_2,(\tau_2, \mathbf{v}_2),\mathbf{w}_2,(\mathbf{B}_2,\check{\mathbf{B}}_2),\mathbf{b}^{\prime}_2,\check{\mathbf{t}}^{\prime}_2,\mathbf{w}^{\prime}_2,\{\mathsf{attr_i}\}_{i\notin\mathcal{I}}^2)$,
			$\mathbf{m}_1=\mathbf{m}_2$ but $\mathbf{y}_1\neq \mathbf{y}_2$, then $\mathcal{C}$ outputs 0. The collision-resistance property of $\mathsf{F}$ implies that $|\Pr_{3}-\Pr_2|\leq negl(\lambda)$.
			\item[{$\mathcal{G}_4$}] is same with $\mathcal{G}_3$, with the difference that during the verification of $\mathcal{A}$’s success, if for two extracted witness
			$(\mathbf{k}_1,\mathbf{y}_1,\mathbf{m}_1,(\tau_1, \mathbf{v}_1),\mathbf{w}_1,(\mathbf{B}_1,\check{\mathbf{B}}_1),\mathbf{b}^{\prime}_1,\check{\mathbf{t}}^{\prime}_1,\mathbf{w}^{\prime}_1,\{\mathsf{attr_i}\}_{i\notin\mathcal{I}}^1)$
			and $(\mathbf{k}_2,\mathbf{y}_2,\mathbf{m}_2,(\tau_2, \mathbf{v}_2),\mathbf{w}_2,(\mathbf{B}_2,\check{\mathbf{B}}_2),\mathbf{b}^{\prime}_2,\check{\mathbf{t}}^{\prime}_2,\mathbf{w}^{\prime}_2,\{\mathsf{attr_i}\}_{i\notin\mathcal{I}}^2)$,
			$\mathbf{y}_1=\mathbf{y}_2$ but $\mathbf{k}_1\neq\mathbf{k}_2$, then $\mathcal{C}$ outputs 0. The uniqueness of the $\mathsf{PRF}$ implies that $|\Pr_4-\Pr_3|\leq negl(\lambda)$. 
			Also, note that $\mathcal{A}$ only utilizes at most $\#AG$  different $usk$s (with $\#AG\leq\|\mathcal{LIST}\|$ and $\#AG$ denotes the number of users in access group of AP) if it succeeds in this game.
			\item[{$\mathcal{G}_5$}] is same with $\mathcal{G}_4$, with the difference that during the verification of $\mathcal{A}$’s success, if for an extracted witness from $\mathcal{LOG}_{AP}$, $(\mathbf{B},\check{\mathbf{B}})$ is not contained in the tag bases set $\mathcal{B}_{AP}$ associated with $AP$, $\mathcal{C}$ outputs 0. The collision resistance property of the static accumulator scheme and the SIS assumption imply that $|\Pr_5-\Pr_4|\leq negl(\lambda)$.
		\end{enumerate}
		
		In $\mathcal{G}_5$, if $\mathcal{A}$ succeeds, it can only use at most $k_{AP}$ distinct tag bases for every $\mathcal{LOG}_{AP}$, and can only use at most $\#AG_\mathcal{}$ distinct $usk$s, therefore $\mathcal{A}$ is only able to generate tags $(\mathbf{t},\check{\mathbf{t}})$ with at most $k_{AP}\times\#AG_\mathcal{}$ distinct $\mathbf{t}$. 
		Hence, it can succeed in $\mathcal{G}_5$ with only negligible advantage, i.e. $\Pr_5\leq negl(\lambda)$. This concludes the proof of D-Detectability of our scheme.\\
		
		\noindent\textbf{\textsl{D-Exculpability for users.}} 
		The proof of D-Exculpability for users for our dynamic $k$-TAA scheme is based on the games defined as follows.
		In addition, we denote a pair of entries $(\mathbf{t}_1,\check{\mathbf{t}}_1,c_1,\mathfrak{m}_1,\Pi_1)$, $(\mathbf{t}_2,\check{\mathbf{t}}_2,c_2,\mathfrak{m}_2,\Pi_2)$ in an $\mathcal{LOG}_{AP}$ as ``malicious pair'' if the pair results in the user $i^*$ being accused, that is, $\mathbf{t}_1=\mathbf{t}_2$, $c_1\neq c_2$ and $(c_1-c_2)^{-1}\cdot(\check{\mathbf{t}}_1-\check{\mathbf{t}}_2)=\mathbf{y}^*$, where $\mathbf{y}^*$ is the $upk$ of user $i^*$. 
		In addition, we refer to an entry in an $\mathcal{LOG}_{AP}$ as a ``malicious entry'' if it belongs to a ``malicious pair'' but is not produced by $i^*$ itself.
		\begin{enumerate}
			\item[$\mathcal{G}_0$] is the real game of \textbf{D-Exculpability for users} defined in \ref{Formal_Definitions_Requirements}.
			\item[$\mathcal{G}_1$] is same with $\mathcal{G}_0$, with the difference that the challenger $\mathcal{C}$ simulate a proof $\Pi$ when it needs to perform authentication representing $i^*$. The zero-knowledge property of the ZKAoK systems implies that $|\Pr_1-\Pr_0|\leq negl(\lambda)$. 
			\item[$\mathcal{G}_2$] is same with $\mathcal{G}_1$, with the difference that during the verification of $\mathcal{A}$’s success, $\mathcal{C}$ also tries to extract the witness  $(\mathbf{k},\mathbf{y},\mathbf{m},(\tau, \mathbf{v}),\mathbf{w},(\mathbf{B},\check{\mathbf{B}}),\mathbf{b}^{\prime},\check{\mathbf{t}}^{\prime},\mathbf{w}^{\prime},\{\mathsf{attr_i}\}_{i\notin\mathcal{I}})$ from every proof $\Pi$ in a ``malicious entry''. If $\mathcal{C}$ is unable to extract a valid witness from any proof $\Pi$ in a ``malicious entry'', it outputs 0, which indicates $\mathcal{A}$ fails. The soundness property of the ZKAoK systems and the fact that each ``malicious entry'' contains a fresh statement-proof pair imply that $| \Pr_2-\Pr_1|\leq negl(\lambda)$.
			\item[$\mathcal{G}_3$] is same with $\mathcal{G}_2$, with the difference that during the verification of $\mathcal{A}$’s success, if for an extracted witness  $(\mathbf{k},\mathbf{y},\mathbf{m},(\tau, \mathbf{v}),\mathbf{w},(\mathbf{B},\check{\mathbf{B}}),\mathbf{b}^{\prime},\check{\mathbf{t}}^{\prime},\mathbf{w}^{\prime},\{\mathsf{attr_i}\}_{i\notin\mathcal{I}})$ from $\mathcal{LOG}_{AP}$, $(\mathbf{B},\check{\mathbf{B}})$ is not contained in $\mathcal{B}_{AP}$ associated with the $AP$, $\mathcal{C}$ outputs 0. The collision resistance of the static accumulator scheme and the SIS assumption imply that $|\Pr_3-\Pr_2|\leq negl(\lambda)$.
			\item[$\mathcal{G}_4$] is same with $\mathcal{G}_3$, with the difference that during the verification of $\mathcal{A}$’s success, if there are two extracted witness $(\mathbf{k}_1,\mathbf{y}_1,\mathbf{m}_1,(\tau_1, \mathbf{v}_1),\mathbf{w}_1,(\mathbf{B}_1,\check{\mathbf{B}}_1),\mathbf{b}^{\prime}_1,\check{\mathbf{t}}^{\prime}_1,\mathbf{w}^{\prime}_1,\{\mathsf{attr_i}\}_{i\notin\mathcal{I}}^1)$
			and $(\mathbf{k}_2,\mathbf{y}_2,\mathbf{m}_2,(\tau_2, \mathbf{v}_2),\mathbf{w}_2,(\mathbf{B}_2,\check{\mathbf{B}}_2),\mathbf{b}^{\prime}_2,\check{\mathbf{t}}^{\prime}_2,\mathbf{w}^{\prime}_2,\{\mathsf{attr_i}\}_{i\notin\mathcal{I}}^2)$ from two ``malicious entries'', the two witness formats a ``malicious pair'' and satisfy $\mathbf{k}_1\neq \mathbf{k}_2$, $\mathcal{C}$ outputs 0. Due to the secret key collision-resistance property of $\mathsf{PRF}$\footnote{This property implies that regardless of whether the inputs are same or not, as long as the $usk$s are different, the outputs of $\mathsf{PRF}$ are different with non-negligible probability, which can be derived from uniqueness and strong uniqueness properties, and is formally proved in \cite{Rupeng_wPRF}.}, we obtain $|\Pr_4-\Pr_3|\leq negl(\lambda)$. 
			Also, note that
			if $\mathcal{A}$ succeeds, for every ``malicious pair'' that consists of two ``malicious entries'', the extracted $usk$s $\mathbf{k}_1$ and $\mathbf{k}_2$ are identical, therefore, we obtain $\check{\mathbf{t}}^\prime_1=\check{\mathbf{t}}^\prime_2$ and $\mathbf{y}_1=\mathbf{y}_2$, implying that $\mathbf{y}_1=\mathbf{y}_2=\mathbf{y}^*$. 
			Consequently, in this case, $\mathcal{C}$ is able to extract the $usk$ corresponding to the $\mathbf{y}^*$.
			\item[$\mathcal{G}_5$] is same with $\mathcal{G}_4$, with the difference that during the verification of $\mathcal{A}$’s success, if there exists ``malicious pair'' consisting of two ``malicious entries'', $\mathcal{C}$ outputs 0. Due to the  one-wayness property of $\mathsf{PRF}$\footnote{This property implies that even if $\mathcal{A}$ knows both the input and the output of $\mathsf{PRF}$, it can only recover the $usk$ with negligible probability, which can be derived from the weak pseudorandomness and uniqueness properties, and is formally proved in \cite{Rupeng_wPRF}.}, we obtain $|\Pr_5-\Pr_4|\leq negl(\lambda)$.
			\item[$\mathcal{G}_6$] is same with $\mathcal{G}_5$, with the difference that during the verification of $\mathcal{A}$’s success, 
			if there exists an extracted witness  $(\mathbf{k},\mathbf{y},\mathbf{m},(\tau, \mathbf{v}),\mathbf{w},(\mathbf{B},\check{\mathbf{B}}),\mathbf{b}^{\prime},\check{\mathbf{t}}^{\prime},\mathbf{w}^{\prime},\{\mathsf{attr_i}\}_{i\notin\mathcal{I}})$ from a ``malicious entries'', which forms a ``malicious pair'' together with an entry generated by $\mathcal{C}$, and $usk$s of this pair satisfy $\mathbf{k}\neq \mathbf{k}^*$, $\mathcal{C}$ outputs 0, where $\mathbf{k}^*$ is the $usk$ of user $i^*$. Due to secret key collision-resistance property of $\mathsf{PRF}$, we obtain $|\Pr_6- \Pr_5| \leq negl(\lambda)$.
			\item[$\mathcal{G}_7$] is same with $\mathcal{G}_6$, with the difference that during the verification of $\mathcal{A}$’s success, if there exists ``malicious pair'' that consists of a ``malicious entry'' and an entry produced by $\mathcal{C}$, then $\mathcal{C}$ outputs 0. Due to the one-wayness property of $\mathsf{PRF}$, we obtain $|\Pr_7-\Pr_6|\leq negl(\lambda)$.
		\end{enumerate}
		
		In $\mathcal{G}_7$, if $\mathcal{A}$ succeeds, there must exist a ``malicious pair'' that consists of two entries generated by $\mathcal{C}$. But the probability that such a ``malicious pair'' exists is negligible, because: 
		(1) the user $i^*$ honestly authenticate with any $AP$ with authentication times less than $k_{AP}$ times; 
		(2) all $(\mathbf{B}_j,\check{\mathbf{B}}_j)$ are sampled uniformly, so the probability of obtaining repeated tag bases is negligible; 
		(3) given the same key sampled by the $\mathsf{PRF}.\mathbf{KeyGen}$ algorithm, the probability that distinct uniformly sampled inputs evaluate to the same output is negligible. 
		Therefore, $\mathcal{A}$ succeeds in $\mathcal{G}_7$ with only negligible advantage, i.e. $\Pr_7\leq negl(\lambda)$.
		This concludes the proof of D-Exculpability for users  of our scheme.
		\\\\
		\textbf{\textsl{D-Exculpability for group manager.}} 
		The proof of D-Exculpability for group manager for our dynamic $k$-TAA scheme is based on the games defined as follows.
		The meanings of ``malicious pair'' and ``malicious entry'' are the same as mentioned above.
		\begin{enumerate}
			\item[$\mathcal{G}_0$] is the real game of \textbf{D-Exculpability for the group manager} defined in \ref{Formal_Definitions_Requirements}.
			\item[$\mathcal{G}_1$] is same with $\mathcal{G}_0$, with the difference that during the verification of $\mathcal{A}$’s success, $\mathcal{C}$ also tries to extract the witness $(\mathbf{k},\mathbf{y},\mathbf{m},(\tau, \mathbf{v}),\mathbf{w},(\mathbf{B},\check{\mathbf{B}}),\mathbf{b}^{\prime},\check{\mathbf{t}}^{\prime},\mathbf{w}^{\prime},\{\mathsf{attr_i}\}_{i\notin\mathcal{I}})$ from every proof $\Pi$ in a ``malicious entry''. If $\mathcal{C}$ is unable to extract a valid witness from any proof $\Pi$ in a ``malicious entry'', it returns 0, which indicates $\mathcal{A}$ does not succeed. The soundness property of the ZKAoK systems implies that $|\Pr_1-\Pr_0|\leq negl(\lambda)$.
			\item[$\mathcal{G}_2$] is same with $\mathcal{G}_1$, with the difference that during the verification of $\mathcal{A}$’s success, if for an extracted witness $(\mathbf{k},\mathbf{y},\mathbf{m},(\tau, \mathbf{v}),\mathbf{w},(\mathbf{B},\check{\mathbf{B}}),\mathbf{b}^{\prime},\check{\mathbf{t}}^{\prime},\mathbf{w}^{\prime},\{\mathsf{attr_i}\}_{i\notin\mathcal{I}})$ from $\mathcal{LOG}_{AP}$, $(\mathbf{B},\check{\mathbf{B}})$ is not contained in $\mathcal{B}_{AP}$ associated with the $AP$, $\mathcal{C}$ outputs 0. The collision resistance of the static accumulator scheme and the SIS assumption imply that $|\Pr_2-\Pr_1|\leq$ $negl(\lambda)$.
			\item[$\mathcal{G}_3$] is same with $\mathcal{G}_2$, with the difference that during the verification of $\mathcal{A}$’s success, if there exist two extracted witness $(\mathbf{k}_1,\mathbf{y}_1,\mathbf{m}_1,(\tau_1, \mathbf{v}_1),\mathbf{w}_1,(\mathbf{B}_1,\check{\mathbf{B}}_1),\mathbf{b}^{\prime}_1,\check{\mathbf{t}}^{\prime}_1,\mathbf{w}^{\prime}_1,\{\mathsf{attr_i}\}_{i\notin\mathcal{I}}^1)$
			and $(\mathbf{k}_2,\mathbf{y}_2,\mathbf{m}_2,(\tau_2, \mathbf{v}_2),\mathbf{w}_2,(\mathbf{B}_2,\check{\mathbf{B}}_2),\mathbf{b}^{\prime}_2,\check{\mathbf{t}}^{\prime}_2,\mathbf{w}^{\prime}_2,\{\mathsf{attr_i}\}_{i\notin\mathcal{I}}^2)$ from two ``malicious entries'', the two witness formats a ``malicious pair'' and satisfy $\mathbf{k}_1\neq \mathbf{k}_2$, then $\mathcal{C}$ outputs 0. Due to the secret key collision-resistance property of $\mathsf{PRF}$, we obtain $|\Pr_3-\Pr_2|\leq negl(\lambda)$.
		\end{enumerate}
		
		In $\mathcal{G}_3$, if $\mathcal{A}$ succeeds, the extracted $usk$s $\mathbf{k}_1$ and $\mathbf{k}_2$ are same for each "malicious pair," resulting in $\check{\mathbf{t}}^\prime_1 = \check{\mathbf{t}}^\prime_2$ and $\mathbf{y}_1 = \mathbf{y}_2$, implying that $\mathbf{y}_1 = \mathbf{y}_2 = \mathbf{y}^\prime$, and $\mathbf{y}^\prime = (c_1 - c_2)^{-1} \cdot (\check{\mathbf{t}}_1 - \check{\mathbf{t}}_2)$.
		Since $\mathbf{m} = \mathsf{bin}(\mathbf{F} \cdot \mathsf{bin}(\mathbf{y}^\prime))$ has not been queried to the group manager during the \textbf{join} algorithm, and due to the unforgeability of the signature scheme \textcolor{black}{and the collision-resistance of $\mathsf{F}$}, $\mathcal{A}$ can succeed in $\mathcal{G}_3$ with negligible probability, i.e. $\Pr_3\leq negl(\lambda)$. This concludes the proof for D-Exculpability for the group manager of our scheme.
	\end{proof}
	\section{Comparison of Communication Cost}
	We compare the efficiency between our work and the existing lattice-based $k$-TAA \cite{Rupeng_wPRF} and lattice-based e-cash \cite{e-cash_DLNS20,YAZ+19_ZKP_on_Z} in this section\footnote{The communication cost of \cite{asiacrypt_LLNW17} already reaches approximately 720 \texttt{TB} at the 80-bit security level \cite{YAZ+19_ZKP_on_Z}, making it far from competitive. Therefore, we do not include it in our comparison.}. Since the total communication cost of the schemes is dominated by the \textbf{authentication} algorithm, we estimate the communication cost of the \textbf{authentication} algorithm between our scheme and the schemes in \cite{Rupeng_wPRF,e-cash_DLNS20}. To estimate bit-security of schemes and choose parameters, we use the online LWE-estimator \cite{LWE_estimator}. For fairness, we adopt the same soundness error for the NIZKAoK and ensure the same level of security for the entire scheme. The communication cost is primarily determined by the size of the NIZKAoK proof $\pi$, so we merely focus on comparing the proof size of our scheme with that of the schemes in \cite{Rupeng_wPRF,e-cash_DLNS20}.
	\begin{table*}[t]
		\caption{\centering Witness length and the size of $\mathcal{M}$ of dynamic $k$-TAA in this paper}\label{Witness_length_Ours}
		\centering 
		\scalebox{1}{
			\begin{threeparttable}
				\begin{tabular}{lll} 
					
					\toprule 
					Terms & Witness length & The size of $\mathcal{M}$ \\
					\midrule\smallskip
					$\check{\mathbf{t}}^\prime=\mathsf{PRF}.\mathbf{Eval}(\mathbf{k},\check{\mathbf{B}})$ &
					$n_1+m_D(2n_1+2+k_1)$  &  $m_D(k_1+n_1)$ \\\smallskip
					$\mathbf{y}=\mathsf{PRF}.\mathbf{Eval}(\mathbf{k},\mathbf{A})$ 
					& $n_1+m_D(2+k_1)$ &  $m_Dk_1$ \\\smallskip
					$\mathbf{t}=\mathsf{PRF}.\mathbf{Eval}(\mathbf{k},\mathbf{B})$ 
					& $n_1+m_D(2n_1+1+k_1)$ & $m_D(k_1+n_1)$ \\\smallskip
					$\mathbf{G}_{n',q_1-1}\cdot\mathbf{m}=\mathbf{F}\cdot \mathsf{bin}(\mathbf{y}) \text{ mod }  p$ 
					& $m_D+m_3+m_Dk_p$ &  $m_3+m_Dk_p$ \\\smallskip
					$\mathsf{DACC}.\mathbf{Verify}((\mathbf{U},\mathbf{A}_{AP}),\mathbf{u}_{AP},\tau,\mathbf{w})=1$ & $km_A+lN$ & $km_A+lN$ \\\smallskip
					$\mathsf{SIGN}.\mathbf{Verify}^{\prime}((\mathbf{A}_{GM},\mathbf{B}_{GM},\mathbf{u}),\tilde{\mathbf{m}},(\tau,\mathbf{v}))=1$ & $L^{\mathsf{SIGN}}_{\mathrm{X}}$ \tnote{$\dagger$} & $L^{\mathsf{SIGN}}_\mathcal{M}$ \tnote{$\diamondsuit$} \\\smallskip
					$\mathsf{ACC}.\mathbf{Tverify}_{}^{}(\mathbf{u},\mathbf{b}^{\prime},\mathbf{w}^{\prime})=1$ &  $2L+4l_1L+2l_2L$ & $L + 2l_1L + 2l_2L$ \\\smallskip
					$\check{\mathbf{t}}^{\prime}+c\cdot \mathbf{y}=\check{\mathbf{t}}$ & $2m_D$ & $0$ \\
					\makecell[l]{$\mathbf{G}_{n_{\mathbf{E}},{q_1}-1}\cdot \mathbf{b}^{\prime}=\mathbf{E}_{}\cdot \mathsf{vdec}(\mathsf{M2V}(\mathbf{B}+\check{\mathbf{B}}))\text{ mod } {q_1}$} 
					& $L^{\mathsf{TagBase}}_{\mathrm{X}}$ \tnote{$\star$}
					& $L^{\mathsf{TagBase}}_{\mathcal{M}}$ \tnote{$\heartsuit$}\\
					\bottomrule
				\end{tabular}
				\begin{tablenotes}   
					\footnotesize              
					\item[$\dagger$]          
					$L^{\mathsf{SIGN}}_{\mathrm{X}}=1+k_{q'}+m_1k_{\alpha_1}+m_2k_{\alpha}+m_3+2n$.
					\item[$\diamondsuit$] $L^{\mathsf{SIGN}}_\mathcal{M}=L_{\mathrm{X}}-n-1$.
					\item[$\star$] $L^{\mathsf{TagBase}}_{\mathrm{X}}=\mathfrak{b}\cdot\lambda\cdot\lceil \log{(m_Dn_1 k'\cdot(2^\iota-1)/\mathfrak{b})}\rceil+n_{\mathbf{E}}\lceil \log{(q_1-1)}\rceil+m_Dn_1 k'+3m_Dn_1$.
					\item[$\heartsuit$] $L^{\mathsf{TagBase}}_{\mathcal{M}}=\mathfrak{b}\cdot\lambda\cdot\lceil \log{(m_Dn_1 k'\cdot(2^\iota-1)/\mathfrak{b})}\rceil+n_{\mathbf{E}}\lceil \log{(q_1-1)}\rceil$.
				\end{tablenotes}      
			\end{threeparttable}
		}
	\end{table*}
	
	\begin{table*}[h]
		\caption{\centering Witness length and the size of $\mathcal{M}$ of $k$-TAA in \cite{Rupeng_wPRF}}
		\centering
		\scalebox {1} {
			\begin{tabular}{ll} 
				\toprule
				Terms & Witness length  \\
				\midrule\smallskip
				$\check{\mathbf{t}}^\prime=\mathsf{PRF}.\mathbf{Eval}(\mathbf{s},\check{\mathbf{B}})$  & 
				\makecell[l]{$m_D(2n\delta_{q_1-1}+3\delta_{\beta_1}+4n\delta_{q_1-1}^2+2\delta_{p-1})+2n\delta_{q_1-1}$}  \\ \smallskip
				$\mathbf{y}=\mathsf{PRF}.\mathbf{Eval}(\mathbf{s},\mathbf{A})$ & $2n\delta_{q_1-1} + 3m_D\delta_{\beta_1} + 2m_D\delta_{p-1}$ \\ \smallskip
				$\mathbf{t}=\mathsf{PRF}.\mathbf{Eval}(\mathbf{s},\mathbf{B})$ & \makecell[l]{$2m_Dn\delta_{q_1-1} + 2n\delta_{q_1-1}+ 3m_D\delta_{\beta_1} + 4m_Dn\delta_{q_1-1}^2$} \\ \smallskip
				$\mathsf{SIGN}.\mathbf{Verify}^{}(PK_{CL},\mathbf{y},\sigma)=1$ &
				\makecell[l]{$m''(6\delta_{\beta_2}\ell+6\delta_{\beta_2}+2+2N +6\delta_{p-1})+3L_0\delta_B+2\ell$} \\ \smallskip
				$\mathsf{ACC}.\mathbf{Tverify}_{}^{}(\mathbf{u},\mathbf{b}^{\prime},\mathbf{w})=1$ & $(5l-1)\cdot 2n\delta_{q_1-1}$ \\ \smallskip
				$\gamma\cdot\check{\mathbf{t}}^{\prime}+\gamma\cdot c\cdot \mathbf{y}=\gamma\check{\mathbf{t}}\text{ mod } q$ & $4m_D\delta_{p-1}$\\ 
				\makecell[l]{$\mathbf{G}_{n,q_1-1}\cdot \mathbf{b}^{\prime}$$=\mathbf{E}_{}\cdot \mathsf{bin}(\mathsf{M2V}(\mathbf{B}\|\check{\mathbf{B}})) \text{ mod } q_1$} & $(2m_D+1)\cdot n\delta_{q_1-1}$\\ 
				\bottomrule
			\end{tabular}
		}
		\label{Witness_length_Yang}
	\end{table*}
	
	\begin{table*}[ht]
		\caption{\centering Estimated parameters for the same security of schemes}
		\centering 
		\scalebox {1} {
			\begin{threeparttable}
				\begin{tabular}{l@{\hskip 0.4cm}l@{\hskip 0.4cm}l@{\hskip 0.4cm}l@{\hskip 0.4cm}l@{\hskip 0.4cm}l@{\hskip 0.4cm}l@{\hskip 0.4cm}l@{\hskip 0.4cm}l}
					\toprule
					& \multicolumn{2}{c}{This Paper} & \multicolumn{2}{c}{\cite{Rupeng_wPRF}} &
					\multicolumn{2}{c}{\cite{YAZ+19_ZKP_on_Z}} & \multicolumn{2}{c}{\cite{e-cash_DLNS20}}\\
					\cmidrule(r){2-3}   \cmidrule(r){4-5}  \cmidrule(r){6-7} \cmidrule(r){8-9} 
					Level & 80 $\mathtt{bits}$ & 128 $\mathtt{bits}$ & 80 $\mathtt{bits}$ & 128 $\mathtt{bits}$ & 80 $\mathtt{bits}$ & 128 $\mathtt{bits}$ & 80 $\mathtt{bits}$ & 128 $\mathtt{bits}$ \\
					\midrule
					$n$  			& 180 		& 175  		& 1270 	& 1565 & 1050 & 1050 &355 & 460	\\		
					$\log{q}$ 		& 201		& 301  		& 96 	& 113 & 201 & 310 &181 &211	\\
					$\log{p}$ 		& 11 		& 11  		& 9		& 13 & 105 & 155 &61 & 71	\\
					$\log{\hat{p}}$ & 80 		& 128  		& ---	& --- & 80 & 128 &80 & 128	\\
					$N_\kappa$ 		& 1 		& 1 		& 137	& 219 & 1 & 1 &1 & 1	\\
					$\|\pi\|$ 		& $\approx\text{314 } \mathtt{MB}$ & $\approx \text{758 } \mathtt{MB}$  &$\geq \text{28.59 }\mathtt{TB}$ &$\geq \text{98.13 }\mathtt{TB}$ & $\approx \text{262 } \mathtt{MB}$ & $\approx \text{671 } \mathtt{MB}$ &$\approx \text{21.43 }\mathtt{GB}$& $\approx\text{38.47 }\mathtt{GB}$ \\
					\bottomrule
				\end{tabular}
			\end{threeparttable}
		}
		\label{parameters_comparison}
	\end{table*}
	
	\begin{table}[ht]
		\caption{\centering Selected parameters of this paper for 80-bits security}
		\centering 
		\scalebox{1}{
			\begin{tabular}{l@{\hspace{2.2mm}}l@{\hspace{2.2mm}}l} 
				\toprule
				Parameters & Description & Value \\
				\midrule
				$n$  				& SIS dimension for $\mathsf{SIGN}$ and $\mathsf{DACC}$  	& 180 \\
				$n_1$  				& SIS dimension for other components  	& 32 \\
				$q_0$				& Base prime for moduli					& $2^{10}+7$ \\
				$p$  				& Prime for $\mathsf{PRF}$ 						& $q_0$ \\
				$q_1$  				& Modulus for $\mathsf{PRF}$ 						& $q_0^{19}$ \\ 
				$q$  				& Modulus for other components 			& $q_0^{20}$ \\
				$q^\prime$  		& Tag bound 							& $2^{10}$ \\
				
				$m\;(\text{or }m_1)$  & Commitment trapdoor dimension 				& 22740 \\
				$m^\prime(\text{or }m_2)$  & $n\cdot\lceil\log{q}\rceil$		& 36180 \\
				$m_e$  				& Bit-size of message for $\mathsf{SIGN}$  						& $2^{12}$ \\	
				
				$m_D$ 				& $\mathsf{PRF}$ trapdoor dimension 				& 1360 \\
				$m_S$ 				& $\mathsf{ACC}$ trapdoor dimension 		& 88440 \\
				$m_A$ 				& $\mathsf{DACC}$ trapdoor dimension 		& 18090 \\
				
				
				$\sigma$ 			& Gaussian width for signatures 	& 2761.34 \\
				$\sigma_1$ 			& $\sqrt{\sigma^2+\sigma_2^2}$			& 2761.37 \\
				$\sigma_2$ 			& Gaussian width for commitment	& 12.73 \\	
				$s$ 				& $\mathsf{DACC}$ preimage sampling width & 1224.81 \\	
				$k$ 				& $\lceil \log{q}\rceil$ 				& 201 \\
				$k_1$ 				& $\lceil \log{p}\rceil$ 				& 11 \\
				$k_2$ 				& $\lceil \log{q_1}\rceil$ 				& 191 \\
				
				$N$ 				& AP's member bound 					& $2^{10}$	\\
				$l$ 				& Merkle tree depth ($\log{N}$) 		& 10 \\
				$\iota$ 			& $\mathsf{vdec}(\cdot)$ decomposition size & 10 \\
				$\mathfrak{b}$ 		& Divided witness block size 			& 40 \\
				\bottomrule
			\end{tabular}
		}
		\label{parameters}
	\end{table}
	
	The ZKAoK used in our work is \cite{YAZ+19_ZKP_on_Z}, 
	we give a detailed analysis of proof size as follows.
	
	We rely on the NIZKAoK scheme of \cite{YAZ+19_ZKP_on_Z}, and use the optimized version of \cite{JRS23_ACS}, which result in the proof size $\|\pi_1\|=$
	\[ \begin{aligned}
		&\left(\log{\left(2\hat{p}+1\right)}+\hat{\kappa}+\left(\hat{l}_1+\mathfrak{n}+\mathfrak{l}\right)\cdot\log{q}\right)\cdot N_\kappa+\left(\hat{l}_1 + \mathfrak{n}\right)\cdot\log{q}\\
		&+\left(\left(2\hat{l}_1+2\hat{l_2}+\mathfrak{n}+\mathfrak{l}\right)\cdot\log\left(\sigma_{2}\log\left(2\hat{l}_1+2\hat{l_2}+\mathfrak{n}+\mathfrak{l}\right)\right)\right)\cdot N_\kappa
	\end{aligned}\]
	where $\mathfrak{n}$ is the total length of witness, $\mathfrak{l}$ is the size of $\mathcal{M}$, and $\hat{p}$, $\hat{\kappa}$, $\hat{l}_1$, $\hat{l}_2$, $q$, $N_\kappa$ are all parameters of ZKAoK system. The bit-size of both the witness length and $\mathcal{M}$'s size of every term in $\pi_1$ is shown in Table \ref{Witness_length_Ours}.

	By combining these statements, the total length of witness is:
	\begin{flalign*}
		\hspace{1.5em}
		\mathfrak{n}&= 3n_1+7m_Dn_1+8m_D+3(\lfloor \log{\frac{q_1}{p}}\rfloor+1)m_D+2 m_3\\
		&+m_D\lceil \log{({p}-1)}\rceil +\lceil \log_2{N}\rceil\cdot N +2\ell+(m_S/2)\cdot2\ell &\\
		&+m_A(\lfloor\log{(2\gamma_1)}\rfloor+1)+{n_{\mathbf{E}}\lceil\log{(q_{1}-1)}\rceil}+4n_1\ell\\
		&+m_1(\lfloor\log{(2\lfloor \sigma_1\log{m_1}\rfloor)}\rfloor+1)+2+\lfloor\log{q'}\rfloor+2n&\\
		&+\mathfrak{b}\lambda\cdot\lceil \log{({m_Dn_1\lceil\log{(q_1-1)}\rceil}/\iota\cdot(2^\iota-1)/{\mathfrak{b}})}\rceil\\
		&+m_2(\lfloor\log{(2\lfloor \sigma\log{m_2}\rfloor)}\rfloor+1)+m_Dn_1\lceil\log{(q_1-1)}\rceil/\iota,&
	\end{flalign*}
	
	and the size of $\mathcal{M}$ is:
	\begin{flalign*}
		\hspace{1.5em}\nonumber\mathfrak{l}\;&= 3(\lfloor \log{\frac{q_1}{p}}\rfloor+1)m_D+2m_Dn_1+m_D\lceil \log{({p}-1)}\rceil\\
		&+ m_A(\lfloor\log{(2\gamma_1)}\rfloor+1)+\lceil \log_2{N}\rceil\cdot N+ \ell+2n_1\ell &\\
		\nonumber &+(m_S/2)\cdot2\ell+ \lfloor\log{q'}\rfloor+1+{n_{\mathbf{E}}\lceil\log{(q_{1}-1)}\rceil}\\
		&+\mathfrak{b}\lambda\cdot\lceil \log{({m_Dn_1\lceil\log{(q_1-1)}\rceil}/\iota\cdot(2^\iota-1)/{\mathfrak{b}})}\rceil\\
		&+m_2(\lfloor\log{(2\lfloor \sigma\log{m_2}\rfloor)}\rfloor+1)+n +2m_3\\
		&+m_1(\lfloor\log{(2\lfloor \sigma_1\log{m_1}\rfloor)}\rfloor+1).&
	\end{flalign*}

	At this point, we can estimate the communication cost of our dynamic $k$-TAA system by setting concrete values to the parameters associated with $\|\pi_1\|$. The security of the system is based on the following nine assumptions:
	$$
	\begin{aligned}
		&SIS_{\hat{l}_1,q,\beta_1};SIS_{\hat{l}_1,q,\beta_2};LWE_{\hat{l}_2,q,\alpha};SIS_{n,q,\beta_4};LWR_{n,q_1,p};\\
		&SIS_{n,q,\beta_5};SIS_{n,q_1,\beta_6};SIS_{n',p,\beta_7};SIS_{n,m_A+l,q,\beta_8},
	\end{aligned}
	$$			
	where $\beta_{1}=16\hat{p}\cdot\sqrt{\hat{l}_{1}+\hat{l}_{2}+\mathfrak{n}_{2}}\cdot(\hat{\sigma}_{2}+\hat{p}\cdot\hat{\sigma}_{1})$, $\beta_{2}=16\hat{p}\cdot\sqrt{\hat{l}_{1}+\hat{l}_{2}+\mathfrak{l}_{2}}\cdot(\hat{\sigma}_{2}+\hat{p}\cdot\hat{\sigma}_{1})$, $\alpha=\frac{\sqrt{2\pi}\cdot\hat{\sigma}_{1}}{q}$, $\beta_{4}=\sqrt{1+(\sqrt{m_1}+\sqrt{m_2}+t)^2}\cdot\sqrt{m_1(\sigma_1\log{m_1})^2+m_2(\sigma\log{m_2})^2}+\min{(2\sqrt{m_1}, \sqrt{m_1} + \sqrt{m_3} + t)}\sqrt{m_3} + 1$, $\beta_{5}=\sqrt{m_{S}}$, $\beta_{6}=\sqrt{m_Dn\lceil \log{({q_1}-1)} \rceil/\iota}\cdot{m_Dn\lceil \log{({q_1}-1)} \rceil}/\iota\cdot(2^\iota-1)/\mathfrak{b}$,  $n'=\frac{m_3}{\lceil\log{p}\rceil}$, $\beta_{7}=\sqrt{m_D\log{p}}$, and  $\beta_8=\sqrt{2\gamma_1^2+s^2m+l}$.
	
	More specifically, we set the parameters according to those listed in Table \ref{parameters}. To ensure 80 bits security, we set $\hat{l}_1=7050$, $\hat{l}_2=7020$, $\hat{p}=2^{80}$, $\hat{\kappa}=80$, $N_\kappa=1$. So, we get $\mathfrak{n}=3708551.34$, $\mathfrak{l}=2648575.34$ and $\|\pi_1\|\approx 314$ $\mathtt{MB}$.

	The communication cost of Abstract Stern's Protocol \cite{LLM+16a_Signature} $\pi_2$ is $$\|\pi_2\|=N_\kappa\cdot\tilde{\mathcal{O}}(L\cdot\log{q}),$$ where $L$ is the witness length. The bit-size of $L$ of every term in $\pi_2$ is summarized in Table \ref{Witness_length_Yang}.

	So, after combing these above, we obtain:
	\begin{align*}
		L&=4m_Dn\delta_{{q_1}-1}+6n\delta_{{q_1}-1}+9m_D\delta_{\beta_1}+8m_Dn\delta_{{q_1}-1}^2\\
		&+2\ell+2m_D\delta_{p-1}+3m''\delta_{\beta_2}(2\ell+2)+2m''+2m''N \\
		&+6m''\delta_{p-1}+3L_0\delta_B+(5l-1)\cdot 2n\delta_{q_1-1}+4m_D\delta_{p-1}\\
		&+(2m_D+1)\cdot n\delta_{q_1-1},
	\end{align*}			
	where $m''=2n\lceil\log q\rceil$, $\beta_1=\lfloor\frac{q_1}{2p}\rfloor$, $\beta_2=\sigma\cdot\omega(\log{m''})$ - the infinity norm bound of signatures, $\delta_{B}=\lfloor\log_2{B}\rfloor+1$, and $L_0=\ell+3n+7m''+3m''N+nN$.

	Then the security of the scheme \cite{Rupeng_wPRF} is based on the following four assumptions:
	$$
	\begin{aligned}
		&SIS_{n,q,\beta_1};LWR_{n,q_1,p};SIS_{n,q,\beta_2};SIS_{n,q_1,\beta_3},\\
	\end{aligned}
	$$			
	where $\beta_1=\sigma^2\cdot(2nk)^{1.5}\cdot80$, $\beta_2=\sqrt{2nk}$, $\beta_3=\sqrt{2mnk_1}$, $\sigma=1.6k\cdot\sqrt{n}$, $k_1=\lceil\log{q_1}\rceil$, and $k=\lceil\log{q}\rceil$. To ensure 80 bits security, we can set $n=1270$, $p=2^{9}-3$, $q=2^{96}-17$ , $q_1=27995$, So, we get $L\approx 19125911100$, and the communication cost $\|\pi_2\|\geq \text{28.59 }\mathtt{TB}$, which is about 67,000 times larger than the communication cost of our scheme. 
	
	The proof size for the scheme of \cite{e-cash_DLNS20} is provided in their work, where the underlying ZKAoK framework is \cite{YAZ+19_ZKP_on_Z}. The witness length is:
	\begin{flalign*}
		\hspace{1.5em}&\ell + (\ell+2)m_s\delta_{2\sigma^\prime} + n\delta_{q_s-1} + 2m_s\delta_{2\sqrt{2}\sigma^\prime} + 3m + m\delta_{q-1} &\\
		&+ m\delta_{p-1}+ m\lceil \log{p}\rceil + 2L + m(L+1) + 2mL + m\delta_{q/p-1},
	\end{flalign*}
	
	the size of $\mathcal{M}$ is:
	\begin{flalign*}
		\hspace{1.5em}&\ell + (\ell+2)m_s\delta_{2\sigma^\prime} + n\delta_{q_s-1} + 2m_s\delta_{2\sqrt{2}\sigma^\prime} + m + m\delta_{q-1} &\\
		&+ m\delta_{p-1}+ m\lceil \log{p}\rceil + 2L + m\delta_{q/p-1}.
	\end{flalign*}
	
	Then the security of the scheme \cite{e-cash_DLNS20} is based on the following four assumptions:
	
	$$LWR_{m,m,q,p}; {LWE}_{n_{\mathsf{LTF}},m_{\mathsf{LTF}},q_{\mathsf{LTF}},\alpha};SIS_{n,m,q,2\sqrt{m}};{SIS}_{n,m_s,q_{s},\beta^{\prime}},$$
	
	where $\alpha=\Theta\left(\frac{\sqrt{n_{\mathsf{LTF}}}}{q_{\mathsf{LTF}}}\right),q_{\mathsf{LTF}}=\Theta(n_{\mathsf{LTF}}^{1+1/\gamma})\text{ for constant }\gamma<1$, and $\beta^{\prime}=\mathcal{O}(\sigma^{2}m_{s}^{1/2}(m_{s}+m\delta_{q-1}))$.
	For 80 bits security, as recommended in \cite{e-cash_DLNS20}, we set $n=355$, $p=2^{60}+33$, $q=p^{3}$, $q_s=p^{3}$, and get the communication cost $\|\pi_3\|\approx \text{21.43 }\mathtt{GB}$. 
	
	A detailed comparison of communication costs among this paper and \cite{Rupeng_wPRF,e-cash_DLNS20,YAZ+19_ZKP_on_Z} is presented in Table \ref{parameters_comparison}. 
	We evaluate the communication costs under two different security levels, 80 $ \mathtt{bits}$ and 128 $\mathtt{bits}$. 
	It shows that the communication cost of our scheme is much lower than that of \cite{Rupeng_wPRF} and \cite{e-cash_DLNS20} in the same security level, our scheme is efficient in terms of communication cost. 
	Although our communication cost is sightly higher than that of \cite{YAZ+19_ZKP_on_Z}, our scheme provides stronger security properties and better features: it achieves exculpability, supports dynamic management, and enables attribute-based authentication. In particular, in both e-cash and $k$-TAA, exculpability is a fundamental security property, and Deo et al. \cite{e-cash_DLNS20} showed that the e-cash scheme of \cite{YAZ+19_ZKP_on_Z} fails to provably achieve this property.

	\section{Conclusion}
	In this paper, we constructed the first lattice-based dynamic $k$-TAA, which not only supports limited number of authentication and selective disclosure of credential attributes for fine-grained authentication, but also offers anonymity, accountability, dynamicity and exculpability. In addition, compared with existing lattice-based $k$-TAA, our scheme achieves higher efficiency in terms of communication cost. We also formally proved its security based on standard lattice assumptions, which ensures post-quantum security. Overall, our results show that dynamic $k$-TAA can provide both efficiency and post-quantum security, and we will continue to improve efficiency to ensure practicality.

	\section*{Acknowledgment}
	This work was supported by the National Natural Science Foundation of China (Grant No. 62372103), the Natural Science Foundation of Jiangsu Province (Grant No. BK20231149) and the Jiangsu Provincial Scientific Research Center of Applied Mathematics (Grant No.BK202330020).
	
	\section*{Data Availability}
	Data will be made available on request.

			\appendix
			
			\section{Zero-Knowledge Arguments of Knowledge}\label{secA1}
			We now detail out the ZKAoK that we use to instantiate the protocols from \ref{Dynamic_k-TAA}. Since we propose a construction over $\mathbb{Z}_q$, we employ the framework from \cite{YAZ+19_ZKP_on_Z} to tackle the relations to be proven. The private information is all placed in \emph{witness}, and if some vectors on $\mathbb{Z}_q$ have norm requirements, it needs to be converted to binary/short vector using the technique in \cite{LNSW13_decomposition}. We also use the method of proving relation with hidden matrices, and the method of transforming the LWR relation \cite{YAZ+19_ZKP_on_Z,LLMNW16_Hidden_Matrix}. 

				We remark that our \textbf{authentication} protocol involves equations over different moduli. Under the assumption that $q$ is a prime power and that every other modulus divides $q$, all such equations can be lifted to the modulus $q$. Concretely, given an equation over $\mathbb{Z}_{q'}$ with $q' \mid q$, we can lift it to $\mathbb{Z}_q$ by multiplying both sides by $q/q' \in \mathbb{Z}$. We rely on this lifting property in the subsequent analysis.

			Now, we need to transform various relations to:
			\[\mathcal{R}^{*}\hspace{0mm}=\hspace{0mm}
			\left\{\hspace{0mm}
			\begin{aligned}
				&(\mathbf{P}, \mathbf{z}, \mathcal{M}), (\mathbf{x})\hspace{0mm}\in\hspace{0mm}\bigl(\mathbb{Z}_{q}^{m\times n}\times \mathbb{Z}_{q}^{m}\times ([1,n]^3)^\ell\bigr)\hspace{0mm}\times\hspace{0mm}(\mathbb{Z}_q^n) \hspace{0mm}:\\
				& \mathbf{P}\cdot\mathbf{x}=\mathbf{z}\wedge\forall(h,i,j)\in\mathcal{M},\mathbf{x}[h]=\mathbf{x}[i]\cdot\mathbf{x}[j]
			\end{aligned}
			\right\},
			\]
			
			and analyze the size of witness and the size of $\mathcal{M}$ after the transformation.
			\subsection{ZKAoK of wPRF Preimage and Key}\label{part_one}
			In this section, we give an argument for the weak pseudorandom function constructed implicitly in \cite{Rupeng_wPRF}. In particular, the argument claims knowledge of a key/input pair that evaluates to a public output.
			More precisely, let $\gamma =\frac{q_1}{p}$ is an odd integer, and we need to lift an equation over $\mathbb{Z}_{q_1}$ to an equation over $\mathbb{Z}_q$, we propose a ZKAoK for the following relation:
			$$
			\mathcal{R}_{wPRF}=\left\{
			\begin{aligned}
				(\mathbf{y}),(\mathbf{A},\mathbf{k})&\in(\mathbb{Z}_{p}^{m})\times(\mathbb{Z}_{q_1}^{m\times n}\times\mathbb{Z}_{q_1}^{n}):\\
				&\mathbf{y}= \lfloor \mathbf{A} \cdot \mathbf{k} \rfloor _p \mod p
			\end{aligned}
			\right\}.
			$$
			
			We construct the argument via reducing the relation $\mathcal{R}_{wPRF}$ to an instance of the relation $\mathcal{R}^{*}$. First, we rewrite the equation $\mathbf{y}= \lfloor \mathbf{A} \cdot \mathbf{k} \rfloor _p \mod p$ as follows:
			$$\left\{
			\begin{aligned}
				\mathbf{A}\cdot\mathbf{k} & =  \mathbf{u}\mod q_1 \\
				\lfloor \frac{p}{q_1}\cdot\mathbf{u} \rfloor & =  \mathbf{y} \mod p
			\end{aligned}
			\right..$$
			
			The first equation is a linear equation with hidden matrix. The second equation holds iff each element of the vector $\mathbf{e}=\mathbf{u}-\gamma\cdot\mathbf{y}\in\mathbb{Z}^m$ is in $[0,\gamma]$, and thus can be transformed into a linear equation with short solution. Namely,
			$$\left\{
			\begin{aligned}
				\mathbf{A}\cdot\mathbf{k} & =  \mathbf{u}\mod q_1 \\
				\gamma\cdot\mathbf{y} &= \mathbf{u}-\mathbf{e} \mod q_1
			\end{aligned}
			\right. .$$
			
			%
			First, 
			for $i\in[1,m]$, we also define $\mathbf{a}_i$ as the $i$-th row of $\mathbf{A}$ and define $\mathbf{v}_i$ as the Hadamard product between $\mathbf{a}_i$ and $\mathbf{k}$, i.e., $\mathbf{v}_i[j]=\mathbf{a}_i[j]\cdot\mathbf{k}[j]$ for $j\in[1,n]$.
			
			Then we decompose the vector $\mathbf{e}^{}$ into a binary vector $\bar{\mathbf{e}}^{}$ using the decomposition technique proposed in \cite{LNSW13_decomposition}. Let $k_1=\lfloor \log{\gamma}\rfloor+1$ and $\mathbf{g}_1=(\lfloor \frac{\gamma+1}{2} \rfloor\|\lfloor \frac{\gamma+2}{4} \rfloor\|...\|\lfloor \frac{\gamma+2^{i-1}}{2^i} \rfloor\|...\|\lfloor \frac{\gamma+2^{k_1-1}}{2^{k_1}} \rfloor)$, and $\mathbf{g}_1$ is a row vector. 
			
			Next, we define the gadget matrix $\mathbf{G}_1=\mathbf{I}_{m}\otimes\mathbf{g}_1$, and it satisfies that $\mathbf{G}_1\cdot\bar{\mathbf{e}}=\mathbf{e}$. 
			Also, we define $\mathbf{a}=(\mathbf{a}_1^{\top}\|...\|\mathbf{a}_m^{\top})^{\top}\in\mathbb{Z}_{q_1}^{m\cdot n}$ and define $\mathbf{v}=(\mathbf{v}_1^{\top}\|...\|\mathbf{v}_m^{\top})^{\top}\in\mathbb{Z}_{q_1}^{m\cdot n}$.
			
			Finally, we set 
			$$\mathbf{P}=\frac{q}{q_1}\cdot\begin{pmatrix}
				\mathbf{0} & \mathbf{0} & \mathbf{M} & -\mathbf{I}_m & \mathbf{0}  \\
				\mathbf{0} & \mathbf{0} & \mathbf{0} & \mathbf{I}_m  & -\mathbf{G}_1  
			\end{pmatrix},$$
			$$\mathbf{x}=\begin{pmatrix}
				\mathbf{k}^{\top} & \mathbf{a}^{\top} & \mathbf{v}^{\top} & \mathbf{u}^{\top} & \bar{\mathbf{e}}^{\top} 
			\end{pmatrix}^{\top}, \mathbf{z}=\frac{q}{q_1}\cdot\left(\begin{array}{c|c}
				\mathbf{0} & \gamma\mathbf{y}^{\top}
			\end{array}\right)^{\top},
			$$
			
			where $\mathbf{M}=\mathbf{I}_m\otimes(1,1,...,1)\in\mathbb{Z}_{q_1}^{m\times m\cdot n}$.
			Besides, we define
				$\mathcal{M}_1=\{(i,i,i)\}_{i\in[n+2mn+m+1,n+2mn+m+k_1m]}$, 
				and $\mathcal{M}_2=\{(n+mn+(i-1)\cdot n+j),n+(i-1)\cdot n+j,j\}_{i\in[1,m],j\in[1,n]}$,
			where $\mathcal{M}_1$ indicates that $\bar{\mathbf{e}}$ is a binary vector and $\mathcal{M}_2$ indicates that $\mathbf{v}_i$ as the Hadamard product between $\mathbf{a}_i$ and $\mathbf{k}$, then we set $\mathcal{M}=\mathcal{M}_1\cup\mathcal{M}_2$. In the new relation, the length of the witness is $n+2mn+m+k_1m$, and the size of $\mathcal{M}$ is $k_1m+mn$.
			
			\subsection{ZKAoK of wPRF Preimage, Key and Output}
			In this section, we give an argument for the weak pseudorandom function constructed implicitly in \cite{Rupeng_wPRF}. In particular, we consider the case that the secret key, the input and the output all need to be hidden.
			More precisely, let $p\geq 2$ and $\gamma =\frac{q_1}{p} \in n^{\omega(1)}$ is an odd integer, we propose a ZKAoK for the following relation:
			$$\mathcal{R}_{wPRF}^{\prime}=\left\{\begin{aligned}
				(\perp),(\mathbf{y},\mathbf{A},\mathbf{k})&\in(\perp)\times(\mathbb{Z}_{p}^{m}\times\mathbb{Z}_{q_1}^{m\times n}\times\mathbb{Z}_{q_1}^{n}):\\
				&\mathbf{y}= \lfloor \mathbf{A} \cdot \mathbf{k} \rceil _p \mod p
			\end{aligned}\right\}. $$
			
			We construct the argument via reducing the relation $\mathcal{R}_{wPRF}^{\prime}$ to an instance of the relation $\mathcal{R}^{*}$. The required vectors and matrices are roughly the same as in Section \ref{part_one}, except:
			$$\mathbf{P}=\frac{q}{q_1}\cdot\begin{pmatrix}
				\mathbf{0} & \mathbf{0} & \mathbf{M} & -\mathbf{I}_m & \mathbf{0}   & \mathbf{0}\\
				\mathbf{0} & \mathbf{0} & \mathbf{0} & \mathbf{I}_m  & -\mathbf{G}_1 & -\gamma\mathbf{I}_{m}
			\end{pmatrix},$$
			$$\mathbf{x}=\begin{pmatrix}
				\mathbf{k}^{\top} & \mathbf{a}^{\top} & \mathbf{v}^{\top} & \mathbf{u}^{\top} & \bar{\mathbf{e}}^{\top} & \mathbf{y}^{\top}
			\end{pmatrix}^{\top}, \mathbf{z}=\frac{q}{q_1}\cdot\left(\begin{array}{c|c}
				\mathbf{0} & \mathbf{0}
			\end{array}\right)^{\top},
			$$
			and
				$\mathcal{M}_1=\{(i,i,i)\}_{i\in[n+2mn+m+1,n+2mn+m+k_1m]}$, 
				and $\mathcal{M}_2=\{(n+mn+(i-1)\cdot n+j),n+(i-1)\cdot n+j,j\}_{i\in[1,m],j\in[1,n]}$,
			where $\mathcal{M}_1$ indicates that $\bar{\mathbf{e}}$ is a binary vector and $\mathcal{M}_2$ indicates that $\mathbf{v}_i$ as the Hadamard product between $\mathbf{a}_i$ and $\mathbf{k}$, then we set $\mathcal{M}=\mathcal{M}_1\cup\mathcal{M}_2$. In the new relation, the length of the witness is $n+2mn+2m+k_1m$, and the size of $\mathcal{M}$ is $k_1m+mn$.
			
			\subsection{ZKAoK of wPRF Key and Output}
			In this section, we give an argument for the weak pseudorandom function constructed implicitly in \cite{Rupeng_wPRF}. In particular, we consider the case that the secret key, the input and the output all need to be hidden.
			
			More precisely, let $p\geq 2$ and $\gamma =\frac{q_1}{p} \in n^{\omega(1)}$ is an odd integer, we propose a ZKAoK for the following relation:
			$$\mathcal{R}_{wPRF}^{\prime\prime}=\left\{\begin{aligned}
				(\mathbf{A}),(\mathbf{y},\mathbf{k})&\in(\mathbb{Z}_{q_1}^{m\times n})\times(\mathbb{Z}_{p}^{m}\times\mathbb{Z}_{q_1}^{n}):\\
				&\mathbf{y}= \lfloor \mathbf{A} \cdot \mathbf{k} \rceil _p \mod p
			\end{aligned}\right\} .$$
			
			
			We construct the argument via reducing the relation $\mathcal{R}_{wPRF}^{\prime\prime}$ to an instance of the relation $\mathcal{R}^{*}$. The required vectors and matrices are roughly the same as in Section \ref{part_one}, except:
			$$\mathbf{P}=\frac{q}{q_1}\cdot\begin{pmatrix}
				\mathbf{A} & -\mathbf{I}_m & \mathbf{0}  & \mathbf{0}\\
				\mathbf{0} & \mathbf{I}_m  & -\mathbf{G}_1 & -\gamma\mathbf{I}_{m}
			\end{pmatrix},$$
			$$\mathbf{x}=\begin{pmatrix}
				\mathbf{k}^{\top} & \mathbf{u}^{\top} & \bar{\mathbf{e}}^{\top} & \mathbf{y}^{\top}
			\end{pmatrix}^{\top}, \mathbf{z}=\frac{q}{q_1}\cdot\left(\begin{array}{c|c}
				\mathbf{0} & \mathbf{0}
			\end{array}\right)^{\top},
			$$
			and
			$$
			\begin{aligned}
				\mathcal{M}_1&=\{(i,i,i)\}_{i\in[n+m+1,n+m+k_1m]} \\
			\end{aligned},$$
			
			where $\mathcal{M}_1$ indicates that $\bar{\mathbf{e}}$ is a binary vector, then we set $\mathcal{M}=\mathcal{M}_1$. In the new relation, the length of the witness is $n+2m+k_1m$, and the size of $\mathcal{M}$ is $k_1m$.
			
			\subsection{ZKAoK of Accumulator Value in \cite{LLNW16_static_acc}}
			Yang et al. \cite{YAZ+19_ZKP_on_Z} have shown how to instantiate an argument of knowledge of an accumulator value for the accumulator scheme presented in \cite{LLNW16_static_acc} in zero-knowledge with the proof system of \cite{YAZ+19_ZKP_on_Z}, please refer to \cite{YAZ+19_ZKP_on_Z}. 
			
			\subsection{ZKAoK of Accumulator Value in \cite{ZYH24_dynamic_acc}}
			In this section, we give an argument of knowledge of an accumulator value for the accumulator scheme presented in \cite{ZYH24_dynamic_acc}.
			More precisely, 
			we propose a ZKAoK for the following relation:
			$$						\mathcal{R}_{DACC}=\left\{
			\begin{aligned}
				(\mathbf{A}&,\mathbf{U},\mathbf{u}),(j,\mathbf{y},\mathbf{w})\in(\mathbb{Z}_q^{n\times m}\times \mathbb{Z}_q^{n\times l\cdot N}\times\mathbb{Z}_q^n)\\
				&\times([N]\times\{0,1\}^{l\cdot N}\times\mathbb{Z}_q^m):\;\\
				&\mathbf{u}=\mathbf{Aw}+\mathbf{Uy}\wedge\|\mathbf{w}\|_\infty\leq\gamma_1\wedge\mathbf{y}=\vartheta_{j}^{N}
			\end{aligned}
			\right\}.$$
			
			We construct the argument via reducing the relation $\mathcal{R}_{DACC}$ to an instance of the relation $\mathcal{R}^{*}$. 
			
			First, we define vectors $\boldsymbol{\gamma}=(\gamma_1,\gamma_1,...,\gamma_1)^{\top}\in\mathbb{Z}_{q}^m$, and define $\mathbf{w}^{\prime}=\mathbf{w}+\boldsymbol{\gamma}$. Then we decompose the vector $\mathbf{w}^{\prime}$ into a binary vector $\bar{\mathbf{w}}^{\prime}$ using the decomposition technique proposed in \cite{LNSW13_decomposition}. Let $k=\lfloor \log{(2\gamma_1)}\rfloor+1$ and $\mathbf{g}=(\lfloor \frac{2\gamma_1+1}{2} \rfloor\|\lfloor \frac{2\gamma_1+2}{4} \rfloor\|...\|\lfloor \frac{2\gamma_1+2^{i-1}}{2^i} \rfloor\|...\|\lfloor \frac{2\gamma_1+2^{k-1}}{2^k} \rfloor)$ be a row vector. Also, we define the gadget matrix $\mathbf{G}=\mathbf{I}_{m}\otimes\mathbf{g}$, and it satisfies that $\mathbf{G}\cdot\bar{\mathbf{w}}^{\prime}=\mathbf{w}^{\prime}$.
			
			Finally, we set 
			$$
			\mathbf{P} = \left(
			\begin{array}{c|c}
				\mathbf{A} \mathbf{G} & \mathbf{U} \\
			\end{array}\right),
			$$
			$$\mathbf{x}=\left(
			\begin{array}{c|c}
				\bar{\mathbf{w}}^{\prime\top} & \mathbf{y}^{\top} 
			\end{array}\right)^{\top}, \mathbf{z}^{}=\mathbf{u}+\mathbf{A}\cdot\boldsymbol{\gamma}.
			$$
			
			Besides, we define
			$$
			\begin{aligned}
				\mathcal{M}&=\{(i,i,i)\}_{i\in[1,km+lN]},
			\end{aligned}
			$$
			
			where $\mathcal{M}$ indicates that $\bar{\mathbf{w}}^{\prime}$ and $\mathbf{y}$ are binary vectors. In the new relation, the length of the witness is $km+lN$, and the size of $\mathcal{M}$ is also $km+lN$.
			
			
			\subsection{ZKAoK of Message-Signature Pair}
			Jeudy et al. \cite{JRS23_ACS} have shown how to instantiate an argument of knowledge of a valid message-signature pair for the signature scheme in zero-knowledge with the proof system of \cite{YAZ+19_ZKP_on_Z}, please refer to \cite{JRS23_ACS}. 
			
			\subsection{Other Relations in {$\Pi$}}	
			\subsubsection{ZKAoK of User Public Key Transformation}
			We propose a ZKAoK for the following relation:
			$$\label{R1}						
			\mathcal{R}_{1}=\left\{\begin{aligned}(\mathbf{G}_{n',p-1}&,\mathbf{F}),(\mathbf{m},\mathbf{y})\in(\mathbb{Z}_{p}^{n'\times m_3}\times\mathbb{Z}_{p}^{n'\times m_D\log{p}})\\
				&\times(\{0,1\}^{m_3}\times\mathbb{Z}_p^{m_D}):\\
				&\mathbf{G}_{n,p-1}\cdot\mathbf{m}=\mathbf{F}\cdot \mathsf{bin}(\mathbf{y}) \mod p\end{aligned}\right\}, 
			$$
			
			where $n' =\frac{m_3}{\lceil\log{p}\rceil}$.
			Then we construct the argument via reducing the relation $\mathcal{R}_{1}$ to an instance of the relation $\mathcal{R}^{*}$, and set:
			$$\mathbf{P}=\frac{q}{p}\cdot\begin{pmatrix}
				\mathbf{G}_{n',p-1}&-\mathbf{F} & \mathbf{0}\\
				\mathbf{0} & \mathbf{G}_{m_D,p-1} & -\mathbf{I}_{m_D}
			\end{pmatrix},$$
			$$\mathbf{x}=\left(\mathbf{m}^\top,\mathsf{bin}(\mathbf{y})^\top,\mathbf{y}^\top\right),
			\mathbf{z}=\frac{q}{p}\cdot\left(\begin{array}{c|c}
				\mathbf{0} & \mathbf{0}
			\end{array}\right)^{\top}.
			$$
			
			%
			
			Besides, we define 
			$$\mathcal{M}=\{(i,i,i)\}_{i\in[1, m_3+m_D\cdot\lceil \log{({p}-1)}\rceil]},$$
			
			where $n' =\frac{m_3}{\lceil\log{p}\rceil}$, $\mathcal{M}$ indicates that $\mathbf{m}$ and $\mathsf{bin}(\mathbf{y})$ are binary vectors. In this relation, the length of the witness is $m_D+m_3+m_D\cdot\lceil \log{({p}-1)}\rceil$, and the size of $\mathcal{M}$ is $m_3+m_D\cdot\lceil \log{({p}-1)}\rceil$.
			
			\subsubsection{ZKAoK of User Tag Equation}
			We propose a ZKAoK for the following relation:
			$$\mathcal{R}_{2}=\left\{\begin{aligned}						(\check{\mathbf{t}},c),(\check{\mathbf{t}}^{\prime},\mathbf{y})\in&(\mathbb{Z}_p^{m}\times\mathbb{Z}_p)\times(\mathbb{Z}_p^{m}\times \mathbb{Z}_p^{m}):\\
				&\check{\mathbf{t}}^{\prime}+ c\cdot\mathbf{y}=\check{\mathbf{t}} \mod p  \end{aligned}\right\}.$$
			
			We construct the argument via reducing the relation $\mathcal{R}_{2}$ to an instance of the relation $\mathcal{R}^{*}$, and set:
			$$\mathbf{P}=\frac{q}{p}\cdot\left(\begin{array}{c|c}
				\mathbf{I}_n & \mathbf{I}_n
			\end{array}\right),$$
			$$\mathbf{x}=\left(\begin{array}{c|c}
				\check{\mathbf{t}}^{\prime\top} & \mathbf{y}^{\top}
			\end{array}\right)^{\top}, \mathbf{z}=\frac{q}{p}\cdot\check{\mathbf{t}}.
			$$
			
			%
			%
			%
			In this relation, the length of the witness is $2m$, and the size of $\mathcal{M}$ is 0.

			\subsubsection{ZKAoK of User Tag Base Transformation}\label{fun_note}
			We recall that ${\mathsf{vdec}}(\cdot): \mathbb{Z}_{q_1}^n\rightarrow[{0,2^\iota-1}]^{nk'}$, which is used to obtain the decomposition of each element of a vector. And $\mathbf{G}_{m}^{\prime}=\mathbf{I}_m\otimes(1,2^\iota,...,2^{(k'-1)\cdot\iota})$. Obviously, $\mathbf{G}_{m}^{\prime}\cdot  \mathsf{vdec}(\mathbf{v})=\mathbf{v}$ for any $\mathbf{v}\in\mathbb{Z}_{q_1}^m$. Let ${\mathbf{C}}=\mathbf{B}+\check{\mathbf{B}}$, we know that $\mathsf{M2V}(\mathbf{C})\in\mathbb{Z}_{q_1}^{mn}$, so $\mathsf{vdec}(\mathsf{M2V}(\mathbf{C})\in[{0,2^\iota-1}]^{mn\cdot\lceil \log{({q_1}-1)}\rceil/ \iota}$, and $\mathbf{E}_{}\cdot \mathsf{vdec}(\mathsf{M2V}(\mathbf{C}))\mod q_1\in\mathbb{Z}_{q_1}^{n}$, then $\mathbf{b}^{\prime}=\mathsf{bin}(\mathbf{E}_{}\cdot \mathsf{vdec}(\mathsf{M2V}(\mathbf{C})))\in\{0,1\}^{n\cdot\lceil\log{({q_1}-1)}\rceil}$. We propose a ZKAoK for the following relation:
			$$
			\mathcal{R}_{3}\hspace{0mm}=\hspace{0mm}\left\{
			\begin{aligned}				
				(\mathbf{G}&_{n_{\mathbf{E}},q_1-1},\mathbf{E}),(\mathbf{b}^{\prime},\mathbf{B},\check{\mathbf{B}})\\
				&\in(\mathbb{Z}_{q_1}^{n_{\mathbf{E}}\times n_{\mathbf{E}}\lceil \log{({q_{}}-1)} \rceil}\times\mathbb{Z}_{q_1}^{n_{\mathbf{E}}\times mn\lceil \log{({q_1}-1)} \rceil/\iota})\\
				&\times(\{0,1\}^{n_{\mathbf{E}}\lceil \log{({q_{1}}-1)} \rceil}\times\mathbb{Z}_{q_1}^{m\times n}\times \mathbb{Z}_{q_1}^{m\times n}):\;\\
				&\mathbf{G}_{n_{\mathbf{E}},{q_1}-1}\cdot \mathbf{b}^{\prime}=\mathbf{E}_{}\cdot \mathsf{vdec}(\mathsf{M2V}(\mathbf{B}+\check{\mathbf{B}}))\text{ mod } {q_1} 
			\end{aligned}
			\right\}.
			$$
			
			We construct the argument via reducing the relation $\mathcal{R}_{3}$ to an instance of the relation $\mathcal{R}^{*}$, and set:
			$$\mathbf{P}=\frac{q}{q_1}\cdot\begin{pmatrix}
				\mathbf{G}_{n_{\mathbf{E}},{q}-1} & -\mathbf{E}_{} & \mathbf{0}& \mathbf{0}& \mathbf{0}\\
				\mathbf{0} & \mathbf{G}_{mn}^\prime & -\mathbf{I}_{mn}& \mathbf{0}& \mathbf{0}\\
				\mathbf{0} & \mathbf{0}& -\mathbf{I}_{mn}&\mathbf{I}_{mn}&\mathbf{I}_{mn}
			\end{pmatrix},$$
			
			$$\mathbf{x}=\left(
			\mathbf{b}^{\prime\top}, \mathsf{vdec}(\mathsf{M2V}(\mathbf{C}))^{\top},\mathsf{M2V}(\mathbf{C})^\top,\mathsf{M2V}{(\mathbf{B})}^\top,\mathsf{M2V}{(\check{\mathbf{B}})}^\top
			\right)^{\top},$$
			$$\mathbf{z}=\frac{q}{q_1}\cdot\left(\begin{array}{c|c|c}
				\mathbf{0}&\mathbf{0} &\mathbf{0} 
			\end{array}\right)^\top.$$ 
			
			%
			
			%
			This is a linear equation with short/binary solutions, $\mathbf{b}^{\prime}$ and $\mathsf{vdec}(\mathsf{M2V}(\mathbf{C}))$, and can be reduced to an instance of the relation $\mathcal{R}^{*}$. Here, we use \emph{the fast mode} to argue that $\mathsf{vdec}(\mathsf{M2V}(\mathbf{C}))$ is a short vector, and this leads to an argument for an instance of $\mathcal{R}^{*}$.
			
			Therefore, the length of the witness is $\mathfrak{b}\cdot\lambda\cdot\lceil \log{(mn k'\cdot(2^\iota-1)/\mathfrak{b})}\rceil+n_{\mathbf{E}}\lceil \log{(q_1-1)}\rceil+mn k'+3mn$, and the size of $\mathcal{M}$ is $\mathfrak{b}\cdot\lambda\cdot\lceil \log{(mnk'\cdot(2^\iota-1)/\mathfrak{b})}\rceil+n_{\mathbf{E}}\lceil \log{(q_1-1)}\rceil$, where $k'=\lceil\log(q_1-1)\rceil/\iota$, $\lambda$ is security parameter, and $\mathfrak{b}$ is the number of divided blocks\footnote{To balance the size of the witness/$\mathcal{M}$ and the hardness of the underlying SIS problem, we divide the ``short'' witness into several, say $\mathfrak{b}$, blocks and use the fast mode on each block.}.

			\section{Formal Definitions of Oracles and Security Requirements}
			\subsection{Oracles}\label{Formal_Definitions_Oracles}
			The adversary has access to a number of oracles and can query them according to the description below, to learn about the system and increase his success chance in the attacks. Their formal definitions can be found in \cite{k-TAA_TFS04} and \cite{dynamic_k-taa_NS05}.
			
			$\mathcal{O}_{LIST}$: Suppose there is an identification list $\mathcal{LIST}$ of user identity/public key pairs. This oracle maintains correct correspondence between user identities and user public keys. Any party can query the oracle to view a user's public key. A user or his colluder can request the oracle to record the user's identity and public key to $\mathcal{LIST}$. The GM or his colluder can request the oracle to delete data from $\mathcal{LIST}$.
			
			$\mathcal{O}_{QUERY}$: It is only used once in the definition for anonymity requirement to give the adversary a challenged authentication transcript. Identities of an AP and two honest users, who are in the AP's access group and have not been authenticated by the AP more than the limit, are given to the oracle. It then randomly chooses one of the two identities, executes the authentication protocol between the chosen identity and the AP, and outputs the transcript of the protocol.
			
			Given a user identity, $\mathcal{O}_{JOIN-GM}$ performs the \textbf{Join} protocol as executed by the honest GM and the user. Given an honest user's identity, $\mathcal{O}_{JOIN-U}$ performs the \textbf{Join} protocol between the GM and the user. Given an honest AP's identity and a user identity, $\mathcal{O}_{AUTH-AP}$ makes the AP to execute the \textbf{authentication} protocol with the user. Given an honest user's identity and an AP identity, $\mathcal{O}_{AUTH-U}$ makes the user to perform the \textbf{authentication} protocol with the AP. $\mathcal{O}_{GRAN-AP}$ takes as input an honest AP's identity and a group member's identity and the AP executes the \textbf{Grant} algorithm to grant access to the user. $\mathcal{O}_{REVO-AP}$ takes as input an honest AP's identity and a member of the AP's access group and the AP executes the \textbf{Revoke} algorithm to revoke the user's access right. $\mathcal{O}_{CORR-AP}$ corrupts an AP specified in its input.
			
			\subsection{Security Properties}
			\label{Formal_Definitions_Requirements}
			In this section, we use $\mathcal{A}$ to denotes the adversary and $\mathcal{C}$ to denotes the challenger.
			
			\subsubsection*{\textbf{D-Anonymity.}} A dynamic $k$-TAA is d-anonymous if for any PPT adversary $\mathcal{A}$, the advantage that $\mathcal{A}$ succeeds in the following game is close to $1/2$ with only negligible gap.
			\begin{enumerate}
				\item[1.] In the beginning, $\mathcal{C}$ chooses a bit $b$, two target users $i_0$, $i_1$, and an application provider AP$^*$ with public information $(ID^*,k^*)$, whose access group contains $i_0$ and $i_1$. Then it publish $(i_0,i_1)$ and $(ID^*,k^*)$.
				\item[2.] Then, $\mathcal{A}$ corrupts the group manager, all application providers (including AP$^*$), and all users except the two target users $i_0,i_1$. It also generates a (malicious) group public key $gpk$.
				\item[3.] Then $\mathcal{A}$ is allowed to access the oracle $\mathcal{O}_{List}[\mathcal{U}\cup\{GM\}-\{i_0,i_1\}],\mathcal{O}_{Join-U}[gpk]$, $\mathcal{O}_{Auth-U}[gpk]$, and $\mathcal{O}_{Query}[b,gpk,(i_0,i_1),(ID^*,k^*)]$ multiple times in any order, but with the restriction that it can only access the oracle $\mathcal{O}_{Query}$ one time of the form $(d,*)$ for each queried $d\in\{0,1\}$, and can only access the oracle $\mathcal{O}_{Auth- U}[ gpk]$ $k-1$ times of the form $( i, ( ID^* , M^* ) )$ for each queried $i\in\{i_0,i_1\}$.
				\item[4.] Finally, $\mathcal{A}$ outputs a bit $b^\prime$ and succeeds if $b=b^\prime$.\\
			\end{enumerate}
			
			\subsubsection*{\textbf{D-Detectability.}} A dynamic $k$-times anonymous authentication is d-detectable if for any PPT adversary $\mathcal{A}$, the advantage that $\mathcal{A}$ succeeds in the following game is negligible.
			\begin{enumerate}
				\item[1.]  Initially, $\mathcal{C}$ generates the group key pair $(gpk,gsk)$ by running the setup procedure, and sends $gpk$ to $\mathcal{A}$. The game has the following two stages.
				\item[2.] In the first stage, $\mathcal{A}$ controls all users in the system and is allowed to access the oracle $\mathcal{O}_{GRAN-AP}$, $\mathcal{O}_{REVO-AP}$, $\mathcal{O}_{CORR-AP}$, $\mathcal{O}_{List}[\mathcal{U}]$, $\mathcal{O}_{Join-GM}[gpk,gsk]$, and $\mathcal{O}_{Auth-AP}[gpk]$ multiple times in any order. After that, all authentication logs of all application providers are emptied.
				\item[3.] In the second stage, $\mathcal{A}$ continues the game, but without access to the revoking oracle $\mathcal{O}_{REVO-AP}$. $\mathcal{A}$ wins if he/she can be successfully authenticated by an honest application provider $\mathcal{V}$ with access bound $k$ for more than $k\times \#AG_\mathcal{V}$ times, where $\#AG_\mathcal{V}$ is the number of members in the AP's access group. The detectability condition requires that the probability that $\mathcal{A}$ wins is negligible.\\
			\end{enumerate}
			
			\subsubsection*{\textbf{D-Exculpability for users.}} A dynamic $k$-times anonymous authentication is d-exculpable for users if for any PPT adversary $\mathcal{A}$, the advantage that $\mathcal{A}$ succeeds in the following game is negligible.
			\begin{enumerate}
				\item[1.] In the beginning, $\mathcal{C}$ chooses a target users $i^*$ and publish it.
				\item[2.] Then, the adversary $\mathcal{A}$ corrupts the group manager, all application providers, and all users except the target user $i^*.$ It also generates a (malicious) group public key $gpk$.
				\item[3.] Then $\mathcal{A}$ is allowed to access the oracle $\mathcal{O}_{List}[\mathcal{U}\cup\{GM\}-\{i^*\}],\mathcal{O}_{Join-U}[gpk]$, and $\mathcal{O}_{Auth-U}[gpk]$ multiple times in any order.
				\item[4.] Finally, $\mathcal{C}$ checks authentication logs of applications providers. $\mathcal{A}$ succeeds in the game if there exists an application provider AP with public information $({ID}_{AP},k_{AP})$ and public $\log\mathcal{LOG}_{AP}$ such that the result of running the \textbf{PublicTracing} algorithm on input $(gpk,\mathcal{LOG}_{AP},\mathcal{LIST})$ consists $i^*$.\\
			\end{enumerate}
			
			\subsubsection*{\textbf{D-Exculpability for the group manager.}} A dynamic $k$-times anonymous authentication is d-exculpable for the group manager if for any PPT adversary $\mathcal{A}$, the advantage that $\mathcal{A}$ succeeds in the following game is negligible.
			\begin{enumerate}
				\item[1.] In the beginning, $\mathcal{C}$ generates the group key pair $(gpk,gsk)$ by running the setup procedure, and sends $gpk$ to the adversary $\mathcal{A}$.
				\item[2.] Then, $\mathcal{A}$ controls all application providers and all users and is allowed to access the oracle $\mathcal{O}_{List}[\mathcal{U}]$ and $\mathcal{O}_{Join-GM}[gpk,gsk]$ multiple times in any order. Here, whenever $\mathcal{A}$ queries the oracle $\mathcal{O}_{Join-GM}[gpk,gsk]$ with an identity $i$ and a user public key $mpk$, it should also query the oracle $\mathcal{O}_{List}[\mathcal{U}]$ with input $(1,i,mpk)$.
				\item[3.] Finally, $\mathcal{C}$ checks authentication logs of applications providers. $\mathcal{A}$ succeeds in the game if there exists an application provider AP with public information $({ID}_{AP},k_{AP})$ and public $\log\mathcal{LOG}_{AP}$ such that the result of running the \textbf{PublicTracing} algorithm on input $(gpk,\mathcal{LOG}_{AP},\mathcal{LIST})$ consists GM.\\
			\end{enumerate}

		\end{document}